\documentclass[a4paper,11pt]{article}
\pdfoutput=1

\usepackage{jheppub}
\usepackage{natbib}
\usepackage{braket}
\usepackage{xcolor}
\usepackage{textcomp}
\usepackage{dcolumn}
\usepackage{bm}
\usepackage{cleveref}
\usepackage{algorithm}
\usepackage{algpseudocode}
\usepackage[utf8]{inputenc}
\usepackage{tcolorbox}
\usepackage{todonotes}
\usepackage{amssymb,amsmath,amsbsy,amsthm,mathtools}
\usepackage{braket}
\usepackage{graphicx,color}
\usepackage{multirow}
\usepackage{hyperref}
\usepackage{subcaption}
\usepackage{adjustbox}
\usepackage{relsize} 
\usepackage{xcolor}
\usepackage{booktabs}
\usepackage{subcaption}
\usepackage{array}

\newcolumntype{P}[1]{>{\centering\arraybackslash}p{#1}}
\newcolumntype{M}[1]{>{\centering\arraybackslash}m{#1}}

\newtheorem*{rep@theorem}{\rep@title}
\newcommand{\newreptheorem}[2]{%
\newenvironment{rep#1}[1]{%
 \def\rep@title{#2 \ref{##1}}%
 \begin{rep@theorem}}%
 {\end{rep@theorem}}}

\newtheorem{definition}{Definition}[section]

\newtheorem{lemma}{Lemma}[section]
\newtheorem{theorem}{Theorem}[section]
\newreptheorem{theorem}{Theorem}
\newtheorem{corollary}{Corollary}[section]

\newtheorem{assumption}{Assumption}[section]

\theoremstyle{definition}
\newtheorem{remark}{Remark}

\newcommand{\mP}{\mathcal{P}}

\newcommand{\mbS}{\mathbb{S}}

\newcommand{\mbZ}{\mathbb{Z}}

\newcommand{\adj}{\overset{\scalebox{0.5}{adj}}{=} }

\usepackage{graphicx}

\title{On the completeness of contraction map proof method for holographic entropy inequalities}

\author[a,b]{Ning Bao}
\author[a]{Keiichiro Furuya}
\author[a,c]{Joydeep Naskar}

\affiliation[a]{Department of Physics, Northeastern University, Boston, MA, 02115, USA}
\affiliation[b]{Computational Science Initiative, Brookhaven National Laboratory, Upton, NY 11973 USA}
\affiliation[c]{The NSF AI Institute for Artificial Intelligence and Fundamental Interactions, Cambridge, MA, U.S.A.}

\emailAdd{ningbao75@gmail.com}
\emailAdd{k.furuya@northeastern.edu}
\emailAdd{naskar.j@northeastern.edu}

\abstract{The contraction map proof method is the commonly used method to prove holographic entropy inequalities. Existence of a contraction map corresponding to a holographic entropy inequality is a sufficient condition for its validity. But is it also necessary? In this note, we answer that question in affirmative for all linear holographic entropy inequalities with rational coefficients. We show that the pre-image of a non-contraction map is not a hypercube, but a proper cubical subgraph, and show that this manifests as alterations to the geodesic structure in the bulk, which leads to the violation of inequalities by holographic geometries obeying the RT formula.}

\makeatletter
\gdef\@fpheader{}
\makeatother

\begin{document}

\maketitle

\section{Introduction}
Quantum entanglement is a ubiquitous phenomenon of nature appearing in many different contexts of quantum physics, and quantum gravity is no exception. The $AdS/CFT$ correspondence \cite{Maldacena1997} has proven to be a hugely successful tool to understand various aspects of quantum gravity, at least perturbatively. By virtue of this correspondence, the structure of quantum entanglement in a boundary CFT mysteriously encodes the information about the geometry of the dual bulk theory of gravity. A part of that mystery has been unfolded by the Ryu-Takayanagi (RT) formula\cite{Ryu-2006-RTformula}, which relates the entanglement entropy of a subregion $R$ of the boundary CFT to the minimal surface $\gamma_R$ in the bulk homologous to $R$. More precisely,
\begin{equation}\label{eq:RT-formula}
        S(R)=\frac{\text{area}(\gamma_R)}{4G_N},
\end{equation}
where $G_N$ is the Newton's constant. This formula holds up to leading order in $1/G_N$, and for $O(G_N^0)$ corrections, one analogously defines the quantum extremal surface \cite{Engelhardt:2014gca,Akers:2021fut}. However, for the purpose of this paper, we will strictly restrict ourselves to the leading order contribution, i.e., the validity of our statements hold upto $O(1/G_N)$.

It is, however, important to clarify that the validity of the RT formula assumes the existence of a semi-classical bulk dual, and not every quantum state on the boundary CFT admits one. We will therefore only talk about those states that admit a semiclassical dual, and are known as \emph{holographic states}. But can we say anything about which states are holographic? It is generally not known, to the best of our knowledge, what are the sufficient conditions for the existence of a bulk dual, however, we do know some necessary ones. A natural question to ask is that since boundary entanglement is so intricately related to the bulk geometry, what can we learn about the bulk geometry by studying entanglement entropy on the boundary CFT? The authors in \cite{Bao:2015bfa} were the first to ask this question, which led to the formulation of the holographic entropy cone (HEC). A series of interesting works followed this development, see references \cite{He:2023aif,He:2019ttu,Hernandez-Cuenca:2022pst,Fadel:2021urx,Bao:2021gzu,Czech:2021rxe,He:2023rox,Bousso:2025mdp,Czech:2025jnw,Bao-2024-towardscompleteness} for a partial list\footnote{We apologize for providing partial list, that excludes many other interesting developments in this field.}.

The HEC is a rational, polyhedral, convex cone defined by its facets, that are given by a finite set of independent, linear inequalities involving subregion entropies. These inequalities, also known as \emph{facet inequalities} bound the entropic phase space of holographic states. In other words, these inequalities are the independent set of tight holographic entropy inequalities (HEI) that cannot be improved further. The holographic entropy cone is completely known for up to five parties \cite{Bao:2015bfa,HernandezCuenca:2019wgh}, and partially known for six parties \cite{Hernandez-Cuenca:2023iqh}, and barring the exception of two infinite families\cite{Bao:2015bfa,Czech:2023xed,Czech:2024rco,Bao:2024vmy}, largely unknown for higher parties.

The standard method to prove an HEI is the contraction map proof method (which we will review in section \ref{sec:review-cmap}). It is proved in \cite{Bao:2015bfa} that if a contraction map exists for a candidate conjecture inequality, then it is a valid HEI. But it remained an open question if the validity of an HEI implies the existence of a contraction map. In a previous work \cite{Bao-2024-properties}, we discovered several properties of the contraction map. These properties shed light on the rigidity of these maps as they admit several rules to deterministically fill up the entries of the contraction map, and at the same time incorporate the flexibility of these maps in terms of free choices. These properties strongly suggest that the entries of the contraction map is not just some constraint satisfaction tool, but are closely related to the space-like geodesic structure on the bulk slice. Contraction map can be thought of as an encapsulation of the inclusion/exclusion principle to encode bulk subregions. These ideas motivated us further investigate the contraction maps from the perspective of graph theory, which led to the discovery of the equivalence of contraction maps to graph contraction maps between hypercubes and partial cubes \cite{Bao-2024-towardscompleteness}. Exploiting this equivalence, we suggested an algorithmic procedure (albeit, a very slow one) to find all HEIs that have corresponding contraction maps\footnote{In a future work, we will report findings on some new HEIs discovered from partial cubes\cite{BFN:preparation}.}.

All known facet HEIs discovered so far admit a contraction map, and there are no known HEI that does not have a contraction map. This is a strong evidence in favour of the completeness of the contraction map proof method. Moreover, the deterministic fixing of the entries of contraction maps and their geometric interpretation via inclusion/exclusion makes the case of completeness stronger, as false inequalities often demand bulk subregions to be included and excluded simultaneously, leading to the failure of contraction condition. Nevertheless, a formal proof of completeness has been elusive. In this work, we shed light on this topic and give a proof of completeness, upto certain disclaimers. We will clarify them one by one,
\begin{itemize}
    \item Our proof of completeness of contraction map proof method rely on the existence of a graph model corresponding to a smooth bulk geometry (and vice versa), which is achieved by a discrete partitioning of the bulk manifold by RT surfaces homologous to boundary subregions. In such models, the discrete entropy is computed on a graph by virtue of edge weights and a min-cut procedure ensures that it equals to the entanglement entropy of the boundary subregion corresponding to the cut. We expect that any sensible holographic geometry admits such a graph realization defined by entanglement entropies given by the RT formula and we are not aware of any counter-example.\footnote{On the other hand, we have examples of non-holographic theories, where the entanglement entropy follows an area law to obey HEIs \cite{Bao:2015boa,Naskar:2024mzi}.}

    \item Our work is closely derived following ideas of the holographic entropy cone. Therefore, the contraction map completion is proved for linear inequalities with rational coefficients realizable within the framework of HEC, applicable to both facet and non-facet HEIs. We do not rule out the possibility of non-linear inequalities.

    \item Our work borrows some of the lessons learned from the existing HEIs, their contraction maps and associated graph structures. We develop a mathematical machinery to define our physical insights. We take some of those insights, for instance, certain adjacency conditions, as physical requirements to construct our proof. We justify their existence and provide supporting evidence for the same.

\end{itemize}

\subsection*{Main Result and Organization}
The main result of this paper is the following theorem:
\begin{reptheorem}{thm:completeness}
    Given a candidate entropy inequality,
    \begin{equation}
        \sum_{i=1}^l \alpha_i S_{P_i} \geq \sum_{j=1}^r \beta_j S_{Q_j}, \; \alpha_i,\beta_j \in \mbZ_{+},
    \end{equation}
    it is a valid HEI if and only if a contraction map $f:\{0,1\}^M\to \{0,1\}^N$ exists (where $M=\sum_i \alpha_i$ and $N=\sum_j \beta_j$) satisfying homology conditions on the boundary.
\end{reptheorem}

The idea of the proof goes as follows: we assume on the contrary that there exists a true HEI that does not have a contraction map. We show that it is possible to always find a geometry that violates the HEI. A non-contractive map breaks the adjacency between distances on graph and Hamming distance of bitstrings. Geometrically, this manifests as change in the geodesic structures in the bulk, thwarting the bulk smoothness. We show that such contradiction arises due to our incorrect assumption about the non-existence of a contraction map corresponding to the HEI. In fact, we will show that non-contraction maps are always associated with alterations of the geodesic structure of the bulk manifold, such that there always exists a holographic geometry for which the inequality is violated.

The organization of this paper is as follows. In section \ref{sec:review-cmap}, we review the relation between geometry and graph, followed by the contraction map proof method. In section \ref{sec:examples}, we give examples of false inequalities and their failure to form contraction maps and go on discussing the implications of non-contraction maps on the bulk geometry. Moreover, we present an example of how a non-contraction map fails a valid HEI, such as MMI. In this section, we build the physical intuitions towards the proof of completeness. In section \ref{sec:proof}, we develop the mathematical machinery and write a formal proof of completeness. Lastly, we discuss the implications of our results and potential applications in section \ref{sec:discussion}.

\section{Review of the Contraction Map Proof Method}
\label{sec:review-cmap}

A $n$-party holographic entropy inequality (HEI) involving $n$ disjoint regions (and a purifier $O$) $[n+1]:=\{A_1, \cdots, A_n, A_{n+1}=O\}$ can be written as
\begin{equation}
    \sum_{i=1}^l \alpha_i S_{P_i} \geq \sum_{j=1}^r \beta_j S_{Q_j},
\end{equation}
where $\alpha_i$ and $\beta_j$ for $i=1,\cdots, l$ and $j=1,\cdots, r$ are positive integers $\mbZ_{+}$. $P_i$ and $Q_j$ are polychromatic subregions, which are the elements of the proper power set $P(\{A_1, \cdots, A_n\}) \backslash \emptyset$ containing $2^n-1$ possible elements. We can expand the inequality to have a unit coefficient for each term, i.e.,
\begin{equation}\label{eq:genentineq-expand}
    \sum_{u=1}^M  S_{L_u} \geq \sum_{v=1}^N  S_{R_v},
\end{equation}
where $L_u,R_v\in P(\{A_1, \cdots, A_n\}) \backslash \emptyset$ for $\forall u =1,\cdots, M$ and $\forall v =1,\cdots, N$, such that
\begin{equation}
    \sum_{i=1}^l \alpha_i = M, \; \sum_{j=1}^r \beta_j = N.
\end{equation}

\subsection{Graphs from holographic geometries and holographic geometries from graphs}

\begin{figure}[t]
    \centering
    \begin{subfigure}{0.26\textwidth}
        \centering
        \includegraphics[width=\linewidth]{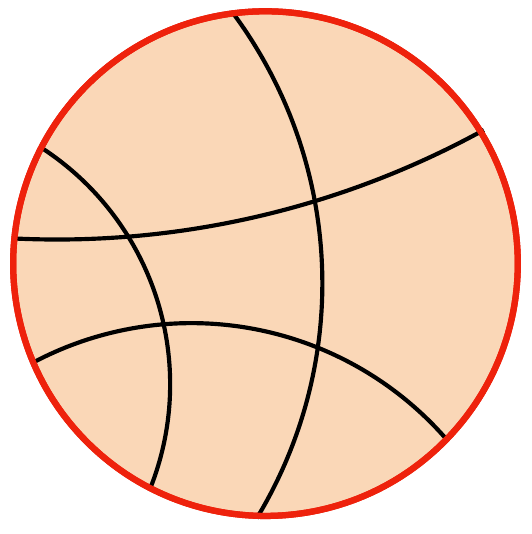}
        \caption{}
        \label{fig:RTarrangement}
    \end{subfigure}\hfill
    \begin{subfigure}{0.26\textwidth}
        \centering
        \includegraphics[width=\linewidth]{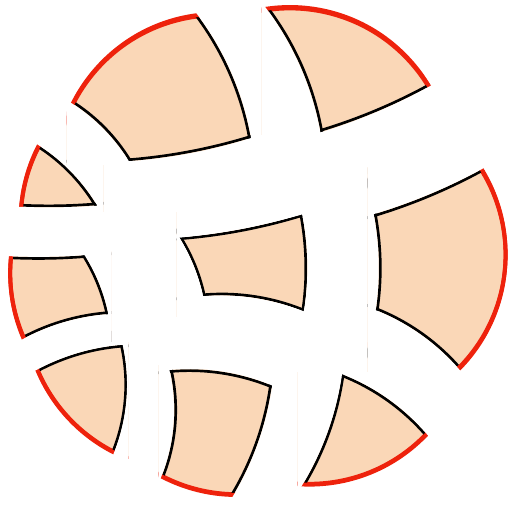}
        \caption{}
        \label{fig:subregions}
    \end{subfigure}\hfill
    \begin{subfigure}{0.26\textwidth}
        \centering
        \includegraphics[width=\linewidth]{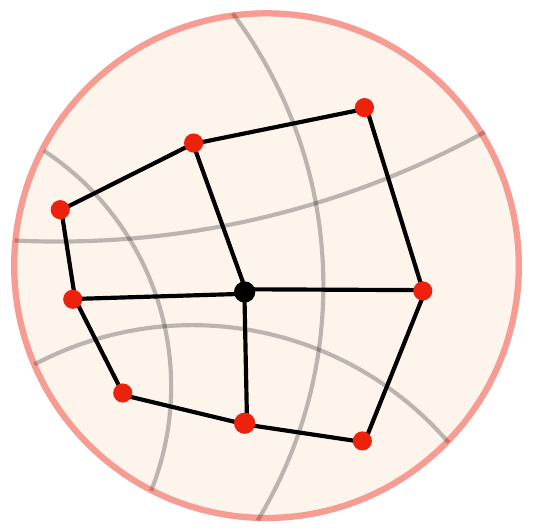}
        \caption{}
        \label{fig:graph}
    \end{subfigure}
    \caption{\small{(a) A holographic geometry with the RT arrangement: Four RT surfaces denoted by black lines are arranged on a constant time slice of AdS$_3$/CFT$_2$. The boundary is colored red. (b) The RT arrangement partitions the time slice into bulk subregions. (c) The partial cube associated with the RT arrangement. The boundary vertices are shown in red. The vertex in the middle, shown in black, is a bulk vertex. Its isometric dimension is $4$, which matches the total number of the RT surfaces. }}
    \label{fig:parititions0}
\end{figure}

For a given set of RT surfaces in a constant time slice of AdS$_{D+1}$/CFT$_{D}$, we have an RT arrangement, defined below, that partitions the bulk manifold into a set of subregions in the bulk, see figure \ref{fig:parititions0}.
\begin{definition}[RT arrangements in AdS$_{D+1}$/CFT$_D$]
    Given a set $\{S_{L_u}| u =1,\cdots,J \}$ of distinct boundary entanglement entropies of $J$ polychromatic subregions $\{L_u \in [n]\}_{u=1}^J$ in AdS$_{D+1}$/CFT$_D$ for some $D\in \mbZ_{\geq 2}$, we have a set of distinct $(D-1)$-dimensional RT hypersurfaces 
    \begin{equation}
        \Gamma_J = \{\gamma_{L_u}| u=1,\cdots,J\}
    \end{equation}
    using the RT formula. Then, the arrangement of RT hypersurfaces, or an RT arrangement, is the set $\Gamma_J$ in a constant time slice of AdS$_{D+1}$. 
\end{definition}

\begin{remark} 
When studying a HEI, we consider two sets of RT arrangements, each of which is associated with either the LHS or the RHS of the inequality. In the later section, we consider the inequality with a unit coefficient by the expansion as in \eqref{eq:genentineq-expand}. For the construction of the corresponding RT arrangements, we distinguish the repetitive terms within each side as the entanglement entropies of distinct subsystems. For example, suppose the LHS of a HEI contains $2 S_{AB}$, which is $S_{AB}+S_{AB}$ by the expansion. We make a distinction between them by making one of them into the entanglement entropy $S_{\widetilde{AB}}$ of the boundary subsystem $\widetilde{AB}$ so that the sum becomes $S_{\widetilde {AB}}+S_{AB}$. Geometrically, one can, for instance, make the boundary subsystem $\widetilde {AB}$ infinitesimally larger than $AB$. In this way, we can avoid a vertex residing in a zero-volume bulk subregion when constructing the RT-region graph defined below from an RT arrangement. Alternatively, one might introduce more number of parties to break the degeneracy of terms such that the inequality has unique terms with unit coefficient. For a discussion along this line, see Corollary 3.1 of \cite{Bao-2024-towardscompleteness}.
\end{remark}

We can construct a graph from the RT arrangement as follows. 
\begin{definition}[RT-region Graph\cite{Bao:2015bfa}]\label{def:RTtoPC}
    Given an RT arrangement $\Gamma_{J}$ in AdS$_{D+1}$/CFT$_D$. The RT-region graph  $G_{J}=(V_{J},E_{J})$ has the following elements.
    \begin{itemize}
        \item Set $V_{J}$ of vertices: Assign a vertex to every bulk subregion
            \begin{itemize}
                \item Boundary vertices $\{v_{A_k}\}\subseteq V_{J}$: A vertex assigned to the bulk subregion ``$a_k$'' homologous to the boundary region $A_k\in[n]$.
                \item Bulk vertices $\{v_b\}\subset V_{J}$; the vertices that are not boundary vertices.
            \end{itemize}
        \item  Set $E_{J}$ of edges: Assign an edge between vertices if the geodesics between them crosses an RT surface once. 
        \begin{itemize}
            \item The edge weights are proportional to the portion of the area of RT hypersurfaces where the edges cross.
        \end{itemize}
    \end{itemize}
    
\end{definition}
By the definition above, the edges of the graph represent the adjacency between the bulk subregions by construction. It can be shown that the underlying unit-weighted graph of the RT-region graph is a special type of graph known as a \textit{partial cube}. We summarize the statement in lemma \ref{lem:RTtoPC}\footnote{The proof can be done by applying the proof for theorem $7.17$ in \cite{Ovchinnikov} to the case of a Riemann manifold of constant negative curvature.}. 

\begin{definition}[Partial cube\cite{Ovchinnikov,WINKLER1984221}]\label{def:partialcube}
    A graph $G=(V,E)$ is a partial cube if $G$ is isometrically embeddable to a $J$-dimensional hypercube graph $H_J=(V_{H_J},E_{H_J})$. That is, there exists an isometry $\phi: V \to V_J$ such that
    \begin{equation}\label{eq:isometric-condition}
            d_G (w,w')=  d_G(\varphi(v),\varphi(v'))
    \end{equation}
    for all $w,w'\in V_J$ and $v,v'\in V_J$.
\end{definition}
\begin{definition}[Isometric dimension \cite{Ovchinnikov}]
    The isometric dimension $idim(G)$ of a partial cube $G$ is the minimum dimension of a hypercube in which $G$ is isometrically embeddable.
\end{definition}
\begin{lemma}\label{lem:RTtoPC}
    Given an RT arrangement $\Gamma_{J}$, the underlying unit-weighted graph of the RT-region graph $G_J$ is a partial cube whose isometric dimension is $idim(G_J)=J$.
\end{lemma}

\begin{remark}
In this work, we mainly use two distance measures, a graph distance and a Hamming distance. They are denoted as $d_G$ and $d_H$, respectively. In this paper, both unit-weighted graphs and weighted graphs appear. When the graph distance is applied to a weighted graph, we measure the distance of its underlying unit-weighted graph, unless explicitly stated otherwise.
\end{remark}

The fact that the underlying unit-weighted graph of an RT-region graph is a partial cube has benefits in studying the HEIs. For example, the isometric embeddability condition reduces the complexity of the algorithm to determine whether a graph is a partial cube or not. It enables the construction of a polynomial-time algorithm to generate the holographic entropy inequalities in \cite{Bao-2024-towardscompleteness}. In this section, we study another characterization of partial cubes, specifically the equivalence class of edges of a partial cube. These equivalence classes have a geometric interpretation. As a result, it simplifies the entropy measure associated with an RT-region graph. We use this simplification in the following sections. 

The edges of a partial cube can be classified into the equivalence classes by the equivalence relation introduced in \cite{WINKLER1984221}. 
\begin{definition}[Winkler (equivalence) relation and its equivalence classes\cite{Ovchinnikov,ovchinnikov2008partial,WINKLER1984221}\footnote{The Winkler relation is a generalization of \textit{Djokovi\'c relation} introduced  in \cite{Djokovi1973Distancepreservings}. When a graph is bipartite, both the Djokovi\'c relation and the Winkler relation coincide. What we call the Winkler relation can also be referred to as the Djoković-Winkler relation. }]\label{def:winkler-relation-class}
    Consider a connected graph $G=(V,E)$. An edge $e=(v_1,v_2)\in E$ for $v_1,v_2\in V$ is in Winkler relation $\Theta$ to another edge $e'=(v'_1,v'_2)\in E$ for $v'_1,v'_2\in V$, which is denoted by $e\Theta e'$, if 
    \begin{equation}\label{eq:parallel-edges}
        d_G(v_1,v_1')+ d_G(v_2,v'_2) \neq d_G(v_1,v_2')+ d_G(v_2,v'_1).
    \end{equation}
    It is an equivalence relation $\Theta$ between edges $e,e'\in E$ if the Winkler relation $\Theta$ is 
    \begin{enumerate}
        \item reflexive, i.e., $e\Theta e$ for $\forall e\in E$, and
        \item symmetric, i.e., $e_1\Theta e_2$ implies $e_2\Theta e_1$ for $\forall e_1,e_2 \in E$, and
        \item transitive, i.e., $e_1\Theta e_2$ and $e_2\Theta e_3$ imply $e_1\Theta e_3$ for $\forall e_1,e_2,e_3\in E$.
    \end{enumerate}
    We denote by $\Theta[e]:=\{e'| e\Theta e', \;\forall e'\in E \}$ the equivalence class of $\Theta$ whose representative edge of the class is $e\in E$.    
\end{definition}
Simply, the edges $e,e'\in E$ in the relation $\Theta$ cannot be a part of the same shortest path in the graph. As discussed in \cite{Ovchinnikov}, consider, for example, a shortest path in a connected graph $G=(V,E)$, and edges $e_1=(v_1,v_2)\in E$ and $e_2=(v_1',v_2')\in E$. If the edges are part of the shortest path and $d_G(v_1,v_2')$ is the largest distance among all possible graph distances between the four vertices, we have
\begin{equation}
    d_G(v_1,v_2') = d_G(v_1,v_1')+1,\;d_G(v_2,v_2') = d_G(v_2,v_1')+1.
\end{equation}
This implies that 
\begin{equation}\label{eq:shortest-path}
     d_G(v_1,v_1')+ d_G(v_2,v_2') = d_G(v_1,v_2') + d_G(v_2,v_1').
\end{equation}

It turns out that the Winkler relation is an equivalence relation on the edge set of any partial cube.
\begin{lemma}[\cite{WINKLER1984221,ovchinnikov2008partial,Ovchinnikov}]\label{lem:partialcube-equivalence re}
    If a connected graph $G=(V,E)$ is a partial cube, then the Winkler relation $\Theta$ is an equivalence relation on $E$. 
\end{lemma}

To understand the relation between the equivalence classes and the RT hypersurfaces, we first introduce \textit{semicubes}.
\begin{definition}[Semicubes\cite{ovchinnikov2008partial}]
    Let $G=(V,E)$ be a connected graph. For any edge $e=(v_1,v_2)\in E$ for $v_1,v_2\in V$, two sets of vertices, $V_e$ and $\overline{V}_e$ are defined as
    \begin{equation}
    \begin{split}
        V_e: = \{v\in V| d_G(v,v_1) <d_G(v,v_2)  \} \\
        \overline{V}_e: = \{v\in V| d_G(v,v_1) >d_G(v,v_2)  \}. \\
    \end{split}
    \end{equation}
    The semicube $G_e$ is the isometric subgraph of $G$ induced by $V_e$. Similarly, the semicube $\overline{G}_e$ is defined accordingly.

\end{definition}
In general, the union of $V_e$ and $\overline{V}_{e}$ does not necessarily add up to $V$. However, if a graph is a partial cube, $V_e \cup \overline{V}_{e} = V$. In addition, the semicubes for any edges in the same equivalence class are equivalent.

\begin{lemma}[\cite{ovchinnikov2008partial}]\label{lem:theta-vertex-partition}
    Let $G=(V,E)$ be a partial cube. Then, 
    \begin{equation}
        e\Theta e', \;e,e'\in E
    \end{equation}
    if and only if the vertex set $V$ is partitioned by $V_e$ and $\overline{V}_e$ or by $V_e'$ and $\overline{V}_e'$ such that
    \begin{equation}\label{eq:vertex-partition}
        V_e= V_{e'}, \text{ and }  \overline{V}_{e} =\overline{V}_{e'},
    \end{equation}
    and 
    \begin{equation}
        G_e = G_{e'}\text{ and }\overline{G}_e = \overline{G}_{e'}
    \end{equation}
\end{lemma}
Lemma \ref{lem:theta-vertex-partition} implies that there is a corresponding RT hypersurface in $\Gamma_J$ for each equivalence class $\Theta$ in $E_J$. This is simply because the bulk subregion on one side of the hypersurface $\gamma$ crossed by an edge $e\in \Theta[e]$ should reside on the same side of the hypersurface $\gamma'$ crossed by any edge $e'\in \Theta[e]$, resulting in $\gamma=\gamma'$.

Hence, for the RT-region graph $G_{J}=(V_{J},E_{J})$ of a given RT arrangement $\Gamma_{J}$, each RT surface $\gamma_{L_u}$ crosses all the edges in one of the equivalence classes $\Theta[e]$ of the RT-region graph. From now on, we label the equivalence class corresponding to the RT surface $\gamma_{L_u}$ as $\Theta_{L_u}[e]$. In addition, we denote the cut due to the equivalence class $\Theta_{L_u}[e]$ by 
\begin{equation}
    V_{J} = V_{L_u}\cup \overline{V}_{L_u}
\end{equation}
instead of $V_{J} = V_e\cup \overline{V}_e$. The vertex subset $V_{L_u}$ contains the boundary vertices $v_{A_k}$ when the boundary subregion $A_k$ is the part of the polychromatic region $L_u$.

The number $|\Gamma_{J}|=J$ of distinct RT surfaces in the RT arrangement is equal to the number $n_\Theta$ of distinct equivalence classes of the RT-region graph. Using the fact\footnote{See theorem 5.34 in \cite{Ovchinnikov}.} that $idim(G_L)=n_\Theta$\cite{Ovchinnikov}, we can see that
\begin{equation}
    |\Gamma_{J}|=idim(G_L) = n_\Theta .
\end{equation}

In general, an entropy measure associated with a weighted graph $G=(V,E)$ can be a discrete entropy\cite{Bao:2015bfa} defined as
\begin{equation}\label{eq:discrete-entropy}
    S^*_l := \underset{W\cup W^c=V}{min} \sum_{(v,v')\in E}|(v,v')|, \; v\in W, \; v'\in W^c
\end{equation}
where its optimization is over the cut $W\cup W^c=V$ of the vertex set $V$. Here, $W \subset V$ is a subset of $V$ and $W^c = V\setminus W$ is a complement set of $W$ in $V$.

It is useful to note the following. First, $V_{J}=V_{L_u}\cup \overline{V}_{L_u}$ for $\Theta_{L_u}[e]$ of an RT-region graph $G_{J}$ defines the cut that gives the minimum of the discrete entropy $S^*_{L_u}$ by construction. Second, the edge weights of all the edges $e$ in the equivalence class $\Theta_{L_u}[e]$ are proportional to the respective portions of area of RT hypersurface $\gamma_{L_u}$. Therefore, we can rewrite the discrete entropy as follows.
\begin{corollary}[]
    Let $G_{J}=(V_{J},E_{J})$ be the RT-region graph of an RT arrangement $\Gamma_{J}$. For $L_u \in [n]$, 
    \begin{equation}\label{eq:discrete-entropy-edges}
        S^*_{L_u} = \sum_{e\in \Theta_{L_u}[e]}|e|
    \end{equation}
    where $|e|$ are the edge weights of $e\in \Theta_{L_u}[e]$.
\end{corollary}

In short, we constructed a pair $(\{S^*_{L_u}\}_{u=1}^J,G_J)$ of the set of discrete entropies and the partial cube from a pair $(\{S_{L_u}\}_{u=1}^J,\Gamma_J)$ of the set of holographic entropies and the set of RT hypersurfaces, such that $S_{L_u}=S^*_{L_u}$ for $\forall u=1,\cdots,J$. Conversely, can we construct a pair $(\{S_{L_u}\}_{u=1}^J,\Gamma_J)$ from $(\{S^*_{L_u}\}_{u=1}^J,G_J)$ such that $S_{L_u}=S^*_{L_u}$ for $\forall u=1,\cdots,J$? The answer is affirmative by construction\footnote{It should be noted that the construction of an RT-arrangement from a pair $(\{S^*_{L_u}\}_{u=1}^J,G_J)$ is not unique because, for example, there can be multiple RT-arrangements $\Gamma_J$ and $\Gamma_J'$ in different spacetime dimensions such that $S_{L_u}=S^*_{L_u}$ for $u=1,\cdots,J$.}. This bidirectional statement is the \textit{geometry-graph duality} formulated in \cite{Bao:2015bfa} from the perspective of partial cubes. Here, we summarize our version of the geometry-graph duality below.
\begin{theorem}[Geometry-graph duality(Partial cubes)]\label{thm:RTPC-PCRT}
    The holographic entropy inequalities can be equivalently defined in terms of
    \begin{enumerate}
        \item pairs $(\{S_{L_u}\}_{u=1}^J,\Gamma_J)$ of the set of holographic entropies and an RT arrangement in AdS$_{D+1}$/CFT$_D$, or
        \item pairs $(\{S^*_{L_u}\}_{u=1}^J,G_J)$ of a set of discrete entropies and an RT-region graph,
    \end{enumerate}
    such that $S_{L_u} = S^*_{L_u}$ for $\forall u$.
\end{theorem}

\subsection{Graphs from bitstrings and bitstrings from graphs}

Every vertex of a $J$-dimensional hypercube $H_J = (V_{H_J},E_{H_J})$ can be labeled by a length $J$ bitstring. By definition \ref{def:partialcube}, every vertex of a partial cube $G_J=(V_J,E_J)$ of $idim(G_J)=J$ can be labeled by a length $J$ bitstring. That is, there exists a map from $V_J$ to $\{0,1\}^J$ while preserving the graph distance between vertices. For later convenience, we explicitly write this down as \textit{isometric condition (graph-to-bitstring)}.
\begin{definition}[Isometric condition w.r.t $H_J$ (graph-to-bitstring)]\label{def:isometric-graph-bitstring}

    A graph $G=(V,E)$ satisfies an isometric condition with respect to a $J$-dimensional hypercube $H_J=(V_{H_J},E_{H_J})$ if there exist an path isometry $\overline{\iota}_J: V_{H_J} \to \{0,1\}^J$ and $\varphi: V \to V_{H_J}$ such that
        \begin{equation}
            d_H (\overline{\iota}_J\circ \varphi (v),\overline{\iota}_J\circ \varphi(v')) = d_G(v,v')
        \end{equation}
        for $\forall v,v'\in V$.  Moreover, we simply write $\overline{\iota}_J\circ \varphi$ as $\overline{\iota}_J: V \to \{0,1\}^J$, and
        \begin{equation}\label{eq:isometric-condition-hypercube}
            d_H (\overline{\iota}_J(v),\overline{\iota}_J(v')) = d_G(v,v').
        \end{equation}
\end{definition}
The path isometries are not unique since any map $\overline{\iota}'_{J}:V \to \{0,1\}^J$ preserving the graph distances as \eqref{eq:isometric-condition-hypercube} can be another path isometry. 

The canonical choice of a path isometry can be determined by the equivalence classes of a partial cube. For example, consider the RT-region graph $G_J=(V_J,E_J)$ of an RT arrangement $\Gamma_J$. First, we define an indicator function $\overline{\iota}_{J}^u: V_J \to \{0,1\}$ for $u=1,\cdots,J$ by
\begin{equation}\label{eq:Winkler-indicator}
    \overline{\iota}_{J}^u(v):=
    \begin{cases}
        1 & v \in V_{L_u}\\
        0 & v \in \overline{V}_{L_u}, \\
    \end{cases}
\end{equation}
where $V_{L_u}$ and $\overline{V}_{L_u}$ are the cut due to the equivalence class $\Theta_{L_u}[e]$. We then construct a path isometry $\overline{\iota}_{J}: V_{J}\to \{0,1\}^J$ by the concatenations of $\overline{\iota}_J^u(v)$, i.e.,
\begin{equation}
    \overline{\iota}_{J}(v) = \overline{\iota}_{J}^1(v) \| \cdots \|\overline{\iota}_{J}^u(v)\|\cdots \| \overline{\iota}_{J}^J(v).
\end{equation}
Here, the parallel lines $\|$ denote the concatenations of bits, e.g., $0\|1 = 01$ for $0,1\in \{0,1\}$ and $01\in \{0,1\}^2$. In particular, we call $\overline{\iota}_{J}(v_{A_k})$ for $\forall A_k \in [n+1]$ as \textit{occurrence} bitstrings.

Physically, the indicator function determines whether the bulk subregion associated with a vertex is inside, $\overline{\iota}_{J}^u(v)=1$, or outside, $\overline{\iota}_{J}^u(v)=0$, the entanglement wedge of the boundary subregion $L_u$ because each equivalence class of the RT-region graph corresponds to an RT hypersurface in the RT-arrangement as discussed in the previous subsection.

We can write the discrete entropy $S^*_{L_u}$ in \eqref{eq:discrete-entropy-edges} using $\overline{\iota}_J$, i.e.,
\begin{equation}\label{eq:discrete-entropy-graph}
\begin{split}
    S^*_{L_u} &= \sum_{(v \in V_{L_u},v'\in \overline{V}_{L_u})\in \Theta_{L_u}[e]}|(v,v')| = \sum_{(v,v')\in E_J} |\overline{\iota}_J^u(v)-\overline{\iota}_J^u(v')|\cdot|(v,v')|\\
\end{split}
\end{equation}
where $|(v,v')|$ is the edge weight of the edge between the vertices $v,v'\in V_J$.

Up to now, we have seen the following;
\begin{enumerate}
    \item (Geometry $\leftrightarrow$ Graph) Given an RT arrangement $\Gamma_{J}$, one can construct the RT-region graph $G_{J}$ and vice versa from theorem \ref{thm:RTPC-PCRT}.
    \item (Geometry $\rightarrow$ Bitstrings) Given an RT arrangement $\Gamma_{J}$, one can construct a set of bitstrings $\iota_{J}(V)$.
\end{enumerate}

We construct an RT arrangement from a set of bitstrings by bypassing theorem \ref{thm:RTPC-PCRT} with a choice of occurrence bitstrings. To see this, we introduce another isometric condition with respect to a $J$-dimensional hypercube. The difference between definition \ref{def:isometric-graph-bitstring} and \ref{def:isometric-bitstring-graph} is whether the path isometry maps from vertices to bitstrings or from bitstrings to vertices.  

\begin{definition}[Isometric condition (bitstring-to-graph)]\label{def:isometric-bitstring-graph}

    Suppose a graph $G_J = (V_J,E_J)$ is induced from a connected subset\footnote{A pair of length $J$ bitstrings, $x,x'\in\{0,1\}^J$, is said to be \textit{connected} if $d_H(x,x')=1$. A sequence of bitstrings $x_0,\cdots, x_n$ forms a \textit{Hamming path} if $d_H(x_i,x_{i+1})=1$ for $\forall i =0, \cdots, n-1$. A \textit{connected subset} $X_{J}\in \{0,1\}^J$ of bitstrings is a subset of $\{0,1\}^J$ such that there is a Hamming path between any pair of elements in $X_{J}$.} of bitstrings $X_J \subseteq \{0,1\}^J$ by a map $\iota_J : \{0,1\}^J \to V_J:=\{\iota_{J}(x)| \forall x\in X_{J} \}$ and $E_J:=\{(\iota_J(x),\iota_J(x'))| x,x'\in X_J\}$.

    The induced graph $G_J$ satisfies the isometric condition if the map is a path isometry, i.e.,
    \begin{equation}\label{eq:dH_dG_notation}
        d_H(x,x') = d_G(\iota_{J}(x),\iota_{J}(x')), \;\forall x,x'\in X_{J}.
    \end{equation}
    For simplicity, we write the condition as 
    \begin{equation}\label{eq:isometric-bitstring-to-graph}
        d_{H}=d_G.
    \end{equation}
\end{definition}

We can construct a partial cube $G_J=(V_J,E_J)$ where
\begin{equation}
\begin{split}
    V_{J} &:= \{\iota_{J}(x)| \forall x\in X_{J} \}\\
    E_{J} &:=\{e=(\iota_{J}(x),\iota_{J}(x'))| \forall x,x'\in X_{J}, \; d_H(x,x')=1\},
\end{split}
\end{equation} from a connected subset $X_J\subseteq \{0,1\}^J$ of bitstrings if and only if the isometric condition $d_H=d_G$ is satisfied. The induced graph $G_J$ is a unit-weighted partial cube. 

To construct an RT arrangement, we first identify occurrence bitstrings for $X_J\subseteq\{0,1\}^J$ and choose edge weights on the induced graph $G_J$. This choice must be consistent with the RT formula, and thus, the geometry is not unique.

Suppose we consider $n+1$ boundary subregions $A_k\in [n+1]$. We choose $n+1$ bitstrings to be occurrence bitstrings $x_{A_k}\in \{0,1\}^J$ for boundary subregions $A_k\in [n+1]$. With the path isometries, we have the boundary vertices $v_{A_k} := \iota_J(x_{A_k})$. This determines the components of polychromatic region $L_u$ for any $u$ and gives the physical labels to the cuts $V_J = V_{L_u}\cup \overline{V}_{L_u}$ or the equivalence classes $\Theta_{L_u}[e]$ of the partial cube $G_J$, i.e.,
\begin{equation} 
\begin{cases}
    \text{ if } x_{A_k}^u = 1, & v_{A_k}\in V_{L_u} \\
    \text{ if } x_{A_k}^u = 0, & v_{A_k}\notin V_{L_u} \\
\end{cases}
\end{equation}
where $x^u_{A_k}$ is the $u$-th bit of the bitstring $x_{A_k} \in \{0,1\}^J$.

For each cut, or equivalence class $\Theta_{L_u}[e]$, we choose a value for $S_{L_u}$. Then, the edge weights are assigned to edges such that \eqref{eq:discrete-entropy-edges} holds.

With the choice of the boundary vertices and the edge weight assignment, we can compute the discrete entropy $S^*_{L_u}$ by
\begin{equation}\label{eq:discrete-entropy-bitstrings}
\begin{split}
    S^*_{L_u} &= \sum_{(\iota_{J}(x) \in V_{L_u},\iota_{J}(x')\in \overline{V}_{L_u})\in \Theta_{L_u}[e]}|(\iota_{J}(x),\iota_{J}(x'))|\\
        &= \sum_{(\iota_{J}(x),\iota_{J}(x'))\in E_{J}}|x_u-x'_{u}|\cdot|(\iota_{J}(x),\iota_{J}(x'))|\\
\end{split}
\end{equation}
where $x_u$ and $x'_u$ are the $u$-th bit of $x$ and $x'$, respectively. Thus, we have a pair $(\{S^*_{L_u}\}_{u=1}^J,G_J)$. Moreover, by construction, we have a pair $(\{S_{L_u}\}_{u=1}^J,\Gamma_J)$ such that $S_{L_u} = S^*_{L_u}$ for $\forall u$.

It should be noted that, for each connected subset $X_J\subset \{0,1\}^J$ of bitstrings with a choice of the occurrence bitstrings, there is a set of pairs $(\{S_{L_u}\}_{u=1}^J,\Gamma_J)$ from $X_J$ because there are choices in the values of the edge weights.

\subsection{Contraction map}

A contraction map is defined as follows.
\begin{definition}[Contraction map]
    A map $f: \{0,1\}^M\to \{0,1\}^N$ is a contraction map if, for $x,x'\in \{0,1\}^M$,
    \begin{equation}
        d_H(x,x') \geq d_H(f(x),f(x')) .
    \end{equation}
\end{definition}

We motivate how the contraction map is related to holographic entropy inequalities by following \cite{Bao:2015bfa}. Then, we review `\textit{proof by contraction}'  method in theorem \ref{thm:proofbycontraction}.

Consider a candidate inequality 
\begin{equation}\label{eq:candidate-hei}
    \sum_{u=1}^M S_{L_u} \geq \sum_{v=1}^N S_{R_v} .
\end{equation}
From theorem \ref{thm:RTPC-PCRT} and the isometric condition $d_H=d_G$, there are connected subsets of bitstrings $X_M\subseteq \{0,1\}^M$ and $Y_N\subseteq \{0,1\}^N$, which contain the occurrence bitstrings
\begin{equation}
    x^u_{A_k} = 
    \begin{cases}
        1, & A_k \in L_u\\
        0, & \text{otherwise}\\
    \end{cases},\;
    y^v_{A_k} = 
    \begin{cases}
        1, & A_k \in R_v\\
        0, & \text{otherwise}\\
    \end{cases},
\end{equation}
where $x^u_{A_k}$ is the $u$-th bit of the occurrence bitstring $x_{A_k} \in \{0,1\}^M$\footnote{With the path isometry, we have $\iota_M(x_{A_k}) =v_{A_k}$ for $\forall k$.}. Similarly, $y^v_{A_k}$ is the $v$-th bit of the occurrence bitstring $y_{A_k} \in \{0,1\}^N$.

Each bitstring $x\in \{0,1\}^M$ represents a bulk subregion consisting a constant time slice of $AdS_{D+1}/CFT_{D}$ partitioned by a set of RT surfaces. The occurrence bitstrings $x_{A_k}\in \{0,1\}^M$ and $y_{A_{k'}} \in \{0,1\}^N$   for any $A_k,A_{k'}\in [n+1]$ represent the bulk subregions that are homologous to the boundary subregion $A_{k}$ and $A_{k'}$. One requires a consistency of the homology conditions between the LHS bitstings and those for the RHS bitstrings under a contraction map, i.e.,
\begin{equation}\label{eq:homology-conditions-bitstrings}
	f(x_{A_k}) = y_{A_k}, \;\forall k \in \{ 1,\cdots, n+1\},
\end{equation}
which imprints the inequality to the contraction map. Then, the contraction map satisfying the consistent homology conditions can be used to prove that the candidate entropy inequality is a valid HEI.
\begin{theorem}[`Proof by contraction' \cite{Bao:2015bfa}; HEI $\leftarrow$ contraction map]\label{thm:proofbycontraction}
    Let $f:\{0,1\}^M \to \{0,1\}^N $ be a contraction map, i.e.,
    \begin{equation} \label{eq:contractioncondition}
        d_H(x ,x') \geq d_H(f(x) ,f(x')) ,\; \forall x,x'\in \{0,1\}^M.
    \end{equation}
    If  $f(x_{A_k})  = y_{A_k}$ for $\forall k\in \{ 1,\cdots, n+1\}$, then (\ref{eq:genentineq-expand}) is a valid $n$-party HEI.
\end{theorem}

\subsection{Graphs and holographic geometries of HEI from contraction map}

In this subsection, we see that theorem \ref{thm:proofbycontraction} is equivalent to the following statement. 
\begin{itemize}
    \item \textit{Any holographic geometry with RT arrangements induced by the contraction map satisfies the HEI. }
\end{itemize}

To see the equivalence, let us assume that there exists a contraction map $f:\{0,1\}^M\to Im(f)\subseteq \{0,1\}^N$ for a candidate inequality \eqref{eq:candidate-hei} satisfying the homology conditions on the boundary $f(x_{A_k})  = y_{A_k}$ for $\forall k\in \{ 1,\cdots, n+1\}$. First, we describe how the contraction map induces holographic geometries. We then demonstrate that all the induced holographic geometries satisfy the inequality.

With the isometric condition $d_H=d_G$ with respect to the $J$-dimensional hypercubes $H_{J}=(V_{H_J},E_{H_J})$ applying to $\{0,1\}^J$ for $J=M,N$, we have 
\begin{enumerate}
    \item path isometries, $\iota_M:\{0,1\}^M \to V_M$, $\iota_N:\{0,1\}^N \to V_N$
    \item the graph contraction map associated with a map $f:\{0,1\}^M \to \{0,1\}^N$ \cite{Bao-2024-towardscompleteness}, 
    \begin{equation}\label{eq:graphmap}
        \Phi_f: H_M \to \Phi_f(H_M) \subseteq H_N .
    \end{equation}
\end{enumerate}
It should be noted that, first, all the graphs above are unit-weighted partial cubes. Second, the image graph $\Phi_f(H_M)$ depends on the contraction map $f$ while the domain of the graph contraction map, which is $H_M$, is the same for any choice of contraction map for the HEI. As we will see later, the \textit{preimage} and the \textit{kernel} of the graph map $\Phi_f$, instead of the domain, are the keys to our proof.

Hence, we now introduce the \textit{preimage} of an image graph under a graph map and the \textit{kernel} of the graph map.

\begin{definition}[Preimage of graph map]\label{def:preimage-graph}
    Let $\Phi:G_1\to G_2$ be a graph map between graphs $G_1=(V_1,E_1)$ and $G_2=(V_2,E_2)$. The preimage 
    \begin{equation}
        PreIm(\Phi(G_1)):=\Big(PreIm_V(\phi^V(V_1)),PreIm_E(\phi^E(E_1))\Big)    
    \end{equation} 
    of $\Phi(G_1) \subseteq G_2$ under a graph map $\Phi:=(\phi^V,\phi^E)$ is
    \begin{equation}
    \begin{split}
        PreIm_V(\phi^V(V_1)):=\{v\in V_1| \phi^V(v) \in V_2 \}  \\
        PreIm_E(\phi^E(E_1)):=\{e\in E_1|\phi^E(e) \in E_2\}. \\
    \end{split}
    \end{equation}
    Note that the preimage $PreIm_E(\phi^E(E_1))$ contains the edges $e\in E_1$ mapped to self-loops. 
\end{definition}

The kernel of a graph map is the preimage of a null subgraph\footnote{A graph is a null graph if its vertex set and edge set are empty\cite{gross-2006-graph}. }.
\begin{definition}[Kernel of graph map]\label{def:kernel-graph}
    The kernel of a vertex map and an edge map of a graph map $\Phi:=(\phi^V,\phi^E)$ on a graph $G:=(V,E)$ is defined as
    \begin{equation}
    \begin{split}
        Ker(\phi^V) := \{v\in V |\phi^V(v) = \emptyset_V \}\\
        Ker(\phi^E) := \{ e\in E |\phi^E(e) = \emptyset_E\}
    \end{split}
    \end{equation}
    where $\emptyset_V$ and $\emptyset_E$ are the null vertices and the null edges, respectively. We denote by $Ker(\Phi):=(Ker(\phi^V),Ker(\phi^E))$ the kernel of the graph map $\Phi$.
\end{definition}
For convenience, we write the set difference of the preimage and the kernel as 
\begin{equation}\label{eq:nontrivial_preimage}
    \mP(\Phi(G)):= PreIm(\Phi(G))\setminus Ker(\Phi).
\end{equation}
We will refer to it as a nontrivial preimage of $\Phi(G)$. 

If a map $f:\{0,1\}^M \to \{0,1\}^N$ is a contraction map, $\mP(\Phi_f(H_M))$ is a $M$-dimensional hypercube $H_M$ because the kernel is trivial, i.e., $Ker(\Phi_f)=(\emptyset_V, \emptyset_E)$. The triviality of the kernel can be simply seen by the fact that the graph contraction map $\Phi_f$ maps the edges of $H_M$ to the edges and the loops in $\Phi_f(H_M)$. Then, most importantly, one can find a corresponding subgraph $G_M\subseteq \mP(\Phi_f(H_M))$, satisfying the isometric condition $d_H=d_G$ \eqref{eq:isometric-bitstring-to-graph}, in the preimage $\mP(\Phi_f(H_M))$ for any connected subset $X_M\in \{0,1\}^M$.

The graph contraction map $\Phi_f$ in \eqref{eq:graphmap} determines the preimage $\mP(\Phi_f(H_M))= H_M$\footnote{If a map is not a contraction map, we will see that   $\mP(\Phi_f(H_M))\subset H_M$ in lemma \ref{lem:kernel-preimage}. We write $\mP(\Phi_f(H_M))= H_M$ explicitly to emphasize that the preimage depends on the contraction map.} and the image graph $\Phi_f(H_M)\subseteq H_N$. With a choice of edge weights, any subgraph $G_M\subseteq \mP(\Phi_f(H_M))$ provides a pair $\{\{S^*_{L_u}\}_{u=1}^{M},G_M\}$ and equivalently $\{\{S_{L_u}\}_{u=1}^{M},\Gamma_M\}$. They give the LHS of the inequality \eqref{eq:candidate-hei}. Similarly, any connected subgraph $G_N\subseteq \Phi_f(H_M)$ with a choice of edge weights provides a pair $\{\{S^*_{R_v}\}_{v=1}^{N},G_N\}$ and equivalently $\{\{S_{R_v}\}_{v=1}^{N},\Gamma_N\}$. 

From theorem \ref{thm:RTPC-PCRT}, we have
\begin{equation}
    \sum_{u=1}^M   S_{L_u} =\sum_{u=1}^M   S^*_{L_u}, \;\sum_{v=1}^N  S_{R_v} =\sum_{v=1}^N  S^*_{R_v}, 
\end{equation}
where the sum of the discrete entropies on the LHS and that of the RHS are
\begin{equation}
\begin{split}
    \sum_{u=1}^M S^*_{L_u} &= \sum_{(\iota_{M}(x),\iota_{M}(x'))\in E_{M}} d_H(x,x')|(\iota_{M}(x),\iota_{M}(x'))|,\; x,x'\in X_M \subseteq \{0,1\}^M\\
    \sum_{v=1}^N S^*_{R_v} &= \sum_{(\iota_{N}(y),\iota_{N}(y'))\in E_{N}} d_H(y,y')|(\iota_{N}(y),\iota_{N}(y'))|,\; y,y'\in f(X_M) \subseteq \{0,1\}^N.\\
\end{split}
\end{equation}

From the assumption, we have a contraction map and its homology conditions \eqref{eq:homology-conditions-bitstrings} for the candidate inequality \eqref{eq:candidate-hei}. Then, we can write down the sum of the discrete entropies on the RHS whose edge weights are not optimized, i.e.,
\begin{equation}
    \sum_{v=1}^N \mbS^*_{R_v}  = \sum_{(\iota_{N}(x),\iota_{N}(x'))\in E_{N}} d_H(f(x),f(x'))|(\iota_{M}(x),\iota_{M}(x'))|.
\end{equation}
In general, we have the inequality\footnote{The inequality \ref{eq:nonoptimized-optimized-inequality} can be proved geometrically based on the entanglement wedge nesting. See figure \ref{fig:mmi-geometry} for an example geometry for the monogamy of mutual information.} 
\begin{equation}\label{eq:nonoptimized-optimized-inequality}
    \sum_{v=1}^N \mbS^*_{R_v}  \geq \sum_{v=1}^N S^*_{R_v}.
\end{equation}
Moreover, we have
\begin{equation}\label{eq:hei-verifier-positive}
\begin{split}
    &\sum_{u=1}^M S^*_{L_u} - \sum_{v=1}^N \mbS^*_{R_v} \\
    &= \sum_{(\iota_M(x),\iota_M(x'))\in E_{M}}\Big( d_H(x,x') - d_H(f(x) ,f(x')) \Big) |(\iota_M(x),\iota_M(x'))| \geq 0\\
\end{split}
\end{equation}
for any choice of edge weights on $G_M$. Therefore, if there exists a contraction map, we have 
\begin{equation}
    \sum_{u=1}^M S^*_{L_u} \geq \sum_{v=1}^N \mbS^*_{R_v}  \geq \sum_{v=1}^N S^*_{R_v}.
\end{equation}

For a contraction map with the homology conditions for the candidate inequality, the same procedure goes through for any choice of edge weights on the induced subgraph $G_M\subseteq \mP(\Phi_f(H_M)) = H_M$. Thus, any induced holographic geometries, or RT arrangements, constructed from $G_M$ with the choices of its edge weights, satisfy the inequality.

Before we move on to the next section, we introduce and discuss an \textit{adjacency condition} defined below. 
\begin{definition}[Adjacency condition]\label{def:adjacency}
    
    Suppose a graph $G_J = (V_J,E_J)$ is induced from a connected subset of bitstrings $X_J \subseteq \{0,1\}^J$ by a map $\iota_J : \{0,1\}^J \to V_J:=\{\iota_{J}(x)| \forall x\in X_{J} \}$ and $E_J:=\{(\iota_J(x),\iota_J(x'))| x,x'\in X_J\}$.

    The induced graph $G_J$ satisfies the adjacency condition if the map satisfies
    \begin{equation}
        d_G(\iota_J(x), \iota_J(x')) = 1
    \end{equation}
     for $x,x'\in X_J$ satisfying
     \begin{equation}
         d_H(x,x') =1.
     \end{equation}
     For simplicity, we write the condition as 
     \begin{equation}\label{eq:adjacencycondition}
         d_H \adj d_G.
     \end{equation}
     
\end{definition}

When a map $\tilde{f}:\{0,1\}^M\to\{0,1\}^N$ is a non-contraction map given in definition \ref{def:non-contractionmap}, we will see in section \ref{sec:proof} that there exists a subgraph $G_M\subset \mP(\Phi_f(H_M))$ in the preimage that does not satisfy the isometric condition $d_H=d_G$ \eqref{eq:isometric-bitstring-to-graph} as well as the adjacency condition \eqref{eq:adjacencycondition}. The map $\tilde{\iota}_M:\{0,1\}^M\to V_M$ that induces the subgraph $G_M$ gives
\begin{equation}
    d_G(\tilde{\iota}_J(x), \tilde{\iota}_J(x')) \neq 1
\end{equation}
for $x,x'\in X_M$ satisfying
\begin{equation}
    d_H(x,x') =1.
\end{equation}

Hence, the induced subgraph does not inherit the adjacency in the connected subset of bitstrings. This discrepancy in adjacency between the two pictures gives rise to obstructions in the smoothness of the geometry induced by RT arrangements. In the next section, we will see that those RT arrangements do not satisfy the inequality.

\section{Lessons from Simple Scenarios}
\label{sec:examples}
In this section, we will discuss three examples to motivate our proof method and illustrate the underlying physics. This will also help us bridge the physical considerations to mathematical statements. 

In particular, when one has a non-contraction map, the image graph is (i) not a partial cube, (ii) a cubical graph, but not a partial cube, and iii) a partial cube. It should be noticed that the image graph being a partial cube is necessary, but not sufficient, for a map to be contractive. The key is that the preimage of the image graph under the graph non-contraction map contains a subgraph that violates the adjacency condition $d_H \adj d_G$ \eqref{eq:adjacencycondition}. As we will see in section \ref{sec:proof}, the existence of such a subgraph leads the non-contraction map to induce geometries that do not satisfy a given inequality. 

In the first and second examples, we present a geometry induced by the non-contraction maps associated with some given false HEIs. The third example deals with the monogamy of mutual information(MMI) with a non-contraction map. We will see that such a map also induces a geometry that does not satisfy the MMI.

\subsection{Deforming the Monogamy of Mutual Information}\label{subsec:mm-deformed}

Consider the MMI inequality
\begin{equation}\label{eq:mmi}
    S(AB)+S(AC)+S(BC) \geq S(A)+S(B)+S(C)+S(ABC),
\end{equation}
for three disjoint regions $A,B,C$ (and a purifier $O$), being probed on a static time-slice of $AdS_3/CFT_2$, such that the associated bulk entanglement wedges $W(AB), W(AC), W(BC)$ and $W(ABC)$ are connected. The contraction map proof method can be thought of as cutting and gluing RT surfaces. Recall the contraction map proof method from section \ref{sec:review-cmap}, see figure \ref{fig:mmi-geometry} for a visual demonstration of the same. The RT surfaces associated with regions $AB$, $BC$, and $CA$ are drawn in red, blue, and green, respectively (fig \ref{fig:lhs-terms}). One can cut the red and blue surfaces and glue them such that they are now homologous to regions $A$ and $C$ (shown in purple and orange dashed curves, respectively) (fig \ref{fig:lhs-cut-glue}), before smoothly deforming them into minimal surfaces homologous to $A$ and $C$ (shown in purple and orange solid curves, respectively). For this particular choice of geometry, the RT surfaces corresponding to $AC$ (drawn in green) is cut and glued as the RT surfaces homologous to $B$ and $ABC$ (fig \ref{fig:rhs-terms}).

\begin{figure}[t]
    \centering
    \begin{subfigure}{0.3\textwidth}
        \centering
        \includegraphics[width=\linewidth]{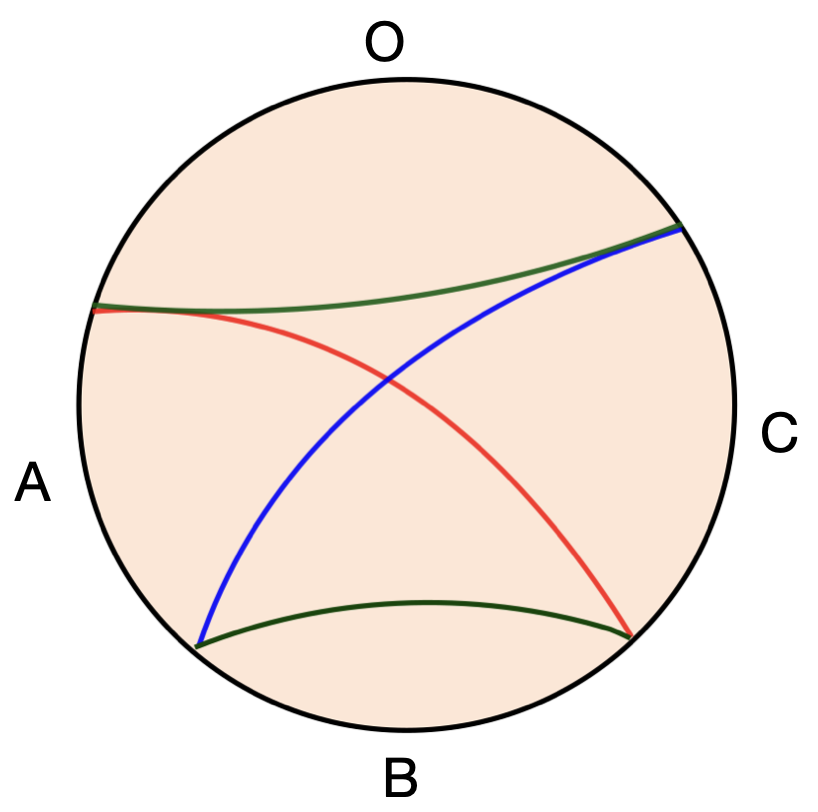}
        \caption{\small{RT surfaces associated with the LHS terms.}}
        \label{fig:lhs-terms}
    \end{subfigure}\hfill
    \begin{subfigure}{0.3\textwidth}
        \centering
        \includegraphics[width=\linewidth]{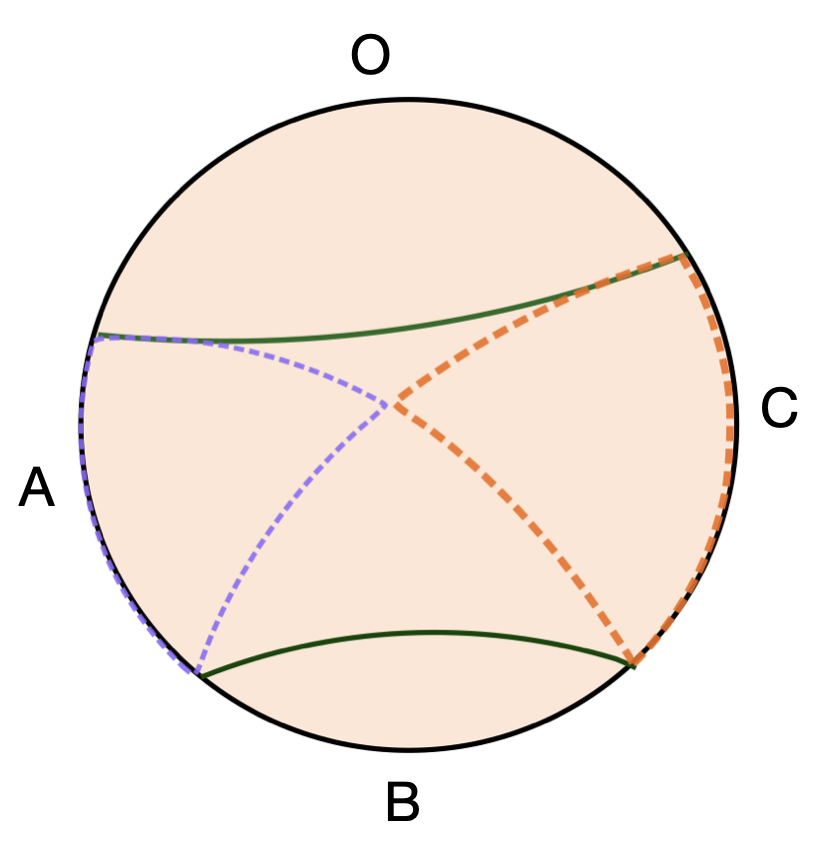}
        \caption{\small{Non-minimal surfaces homologous to regions in RHS from LHS RT surfaces.}}
        \label{fig:lhs-cut-glue}
    \end{subfigure}\hfill
    \begin{subfigure}{0.3\textwidth}
        \centering
        \includegraphics[width=\linewidth]{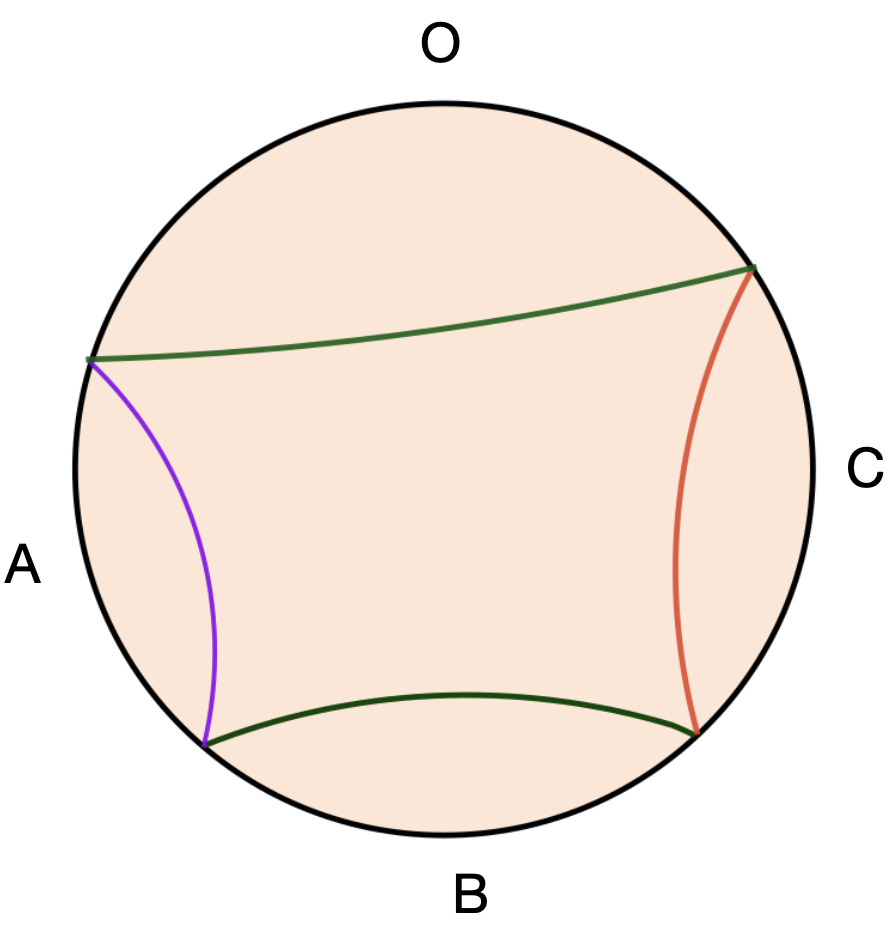}
        \caption{\small{RT surfaces associated with the RHS terms.}}
        \label{fig:rhs-terms}
    \end{subfigure}
    \caption{\small{The RT surfaces associated with the LHS terms (\ref{fig:lhs-terms}) is cut and glued (rearranged) to correspond to (non-minimal) surfaces associated with RHS terms (\ref{fig:lhs-cut-glue}) and finally smoothly deformed to their respective minimal surfaces (\ref{fig:rhs-terms}).}}
    \label{fig:mmi-geometry}
\end{figure}

Let us now deform the MMI inequality to,
\begin{equation}\label{eq:deformed-mmi}
    S(AB)+S(AC)+S(BC) \geq S(A)+S(B)+S(C)+S(ABC)+ \textcolor{red}{S(AB)}.
\end{equation}
Clearly, this is a false inequality. We will now discuss what goes wrong at the level of contraction map (see table \ref{tab:cmap-deformed-mmi}). We append a column to the MMI contraction map corresponding to extra term. We see that after imposing the boundary conditions, we reach a point of conundrum where any bit choice would result in a contradiction, i.e, those bits are simultaneously required to be $0$ and $1$. What we see here is a false inequality leading to a non-contractive mapping. This feature was observed in a previous work \cite{Bao-2024-properties}.

\begin{table}[h!]
\centering
\begin{tabular}{@{}lccclccccc@{}}
\toprule
 & \textbf{AB} & \textbf{AC} & \textbf{BC} &  & \textbf{A} & \textbf{B} & \textbf{C} & \textbf{ABC} & \multicolumn{1}{l}{{\color[HTML]{FE0000} \textbf{AB}}} \\ \cmidrule(r){1-4} \cmidrule(l){6-10} 
\textbf{O} & \textbf{0} & \textbf{0} & \textbf{0} &  & \textbf{0} & \textbf{0} & \textbf{0} & \textbf{0} & {\color[HTML]{FE0000} 0} \\
 & 0 & 0 & 1 &  & 0 & 0 & 0 & 1 & {\color[HTML]{FE0000} } \\
 & 0 & 1 & 0 &  & 0 & 0 & 0 & 1 & {\color[HTML]{FE0000} } \\
\textbf{C} & \textbf{0} & \textbf{1} & \textbf{1} &  & \textbf{0} & \textbf{0} & \textbf{1} & \textbf{1} & {\color[HTML]{FE0000} \textbf{0}} \\
 & 1 & 0 & 0 &  & 0 & 0 & 0 & 1 & {\color[HTML]{FE0000} } \\
\textbf{B} & \textbf{1} & \textbf{0} & \textbf{1} &  & \textbf{0} & \textbf{1} & \textbf{0} & \textbf{1} & {\color[HTML]{FE0000} \textbf{1}} \\
\textbf{A} & \textbf{1} & \textbf{1} & \textbf{0} &  & \textbf{1} & \textbf{0} & \textbf{0} & \textbf{1} & {\color[HTML]{FE0000} \textbf{1}} \\
 & 1 & 1 & 1 & \multirow{-9}{*}{\textbf{}} & 0 & 0 & 0 & 1 & {\color[HTML]{FE0000} } \\ \bottomrule
\end{tabular}
\caption{Failure to fill up the contraction map of inequality \ref{eq:deformed-mmi}. The empty bits are simultaneously required to be 0 and 1 to satisfy the contraction map condition.}
\label{tab:cmap-deformed-mmi}
\end{table}

Using the triality proposed in \cite{Bao-2024-towardscompleteness}, we will now discuss what it means in the graph picture. We will choose the remaining bits (which gives a non-contraction map) and draw the graph (see figure \ref{fig:mmi-deformed}).
\begin{figure}[t]
    \centering
    \begin{subfigure}{0.3\textwidth}
        \centering
        \includegraphics[width=\linewidth]{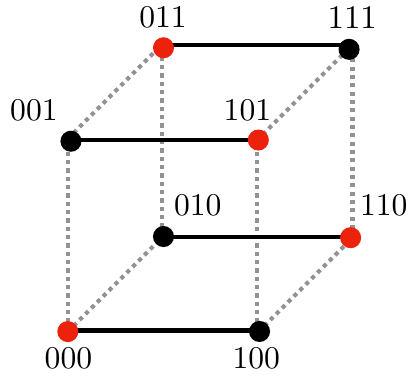}
        \caption{\small{$\mP(\Phi_{\tilde{f}}(H_3))$}}
        \label{fig:mmi_deformed_preimage}
    \end{subfigure}\hfill
    \begin{subfigure}{0.5\textwidth}
        \centering
        \includegraphics[width=\linewidth]{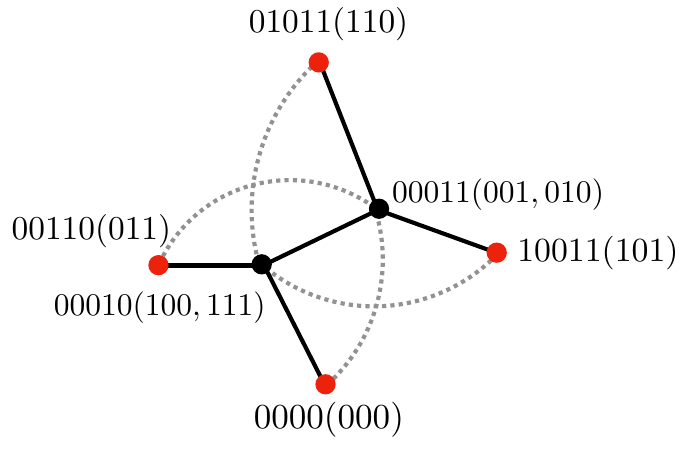}
        \caption{\small{$\Phi_{\tilde{f}}(H_3)$}}
        \label{fig:mmi_deformed_image}
    \end{subfigure}\hfill
    \caption{\small{The violations of $d_H \adj d_G$: The dotted lines emphasize the null edges that violate the adjacency condition. The red circles are the boundary vertices. The black dots are the bulk vertices. (a) The nontrivial preimage $\mP(\Phi_{\tilde{f}}(H_3))$ has eight edges missing. Those eight edges belong to the kernel $Ker(\Phi_{\tilde{f}})$. (b) The image graph $\Phi_{\tilde{f}}(H_3)$  corresponding to the non-contraction map $\tilde{f}$ after choosing the empty bits in the last column to be $1,1,0,0$ for the second, third, fifth, and eighth row, respectively.}}
    \label{fig:mmi-deformed}
\end{figure}
We see that the image graph has an isometric embedding in $H_5$ respecting the boundary conditions and is a partial cube. So the image graph being a partial cube is not a sufficient condition, but rather a necessary one for the validity of a HEI. Looking at the table \ref{tab:cmap-deformed-mmi}, it is clear that the domain of the graph maps to the image without preserving all adjacent edges. Such a mapping does not take a RT arrangement of graphs to another smoothly. Strictly from a graph picture, this can be traced back to deletion of edges in the original hypercube such that the pre-image is a cubical graph instead of a hypercube. By deleting edges, one violates the adjacency condition $d_H\adj d_G$ in \eqref{eq:adjacencycondition} between vertices whose bitstring labels differ by a single bit. Recall that the vertices of the graph in the holographic geometry correspond to bulk subregions, chambered by the RT surfaces. Each such bulk subregion is mapped with a bitstring and two subregions separated by a RT surface differ by a single bit (unless there are zero-volume chambers sitting on the RT surfaces, which are referred as \emph{unphysical bitstrings}). The contraction map can be thought of as an operation on the RT surfaces to be smoothly deformed from one configuration to another. In this example, there is no such continuous deformation of the RT surfaces corresponding to the LHS into those corresponding to the RHS. So, how exactly does this graph mapping differ from those of graph contraction maps \cite{Bao-2024-towardscompleteness} (also known as \emph{weak graph homomorphisms}\cite{Bao:2015bfa})?

In the usual story of graph contraction maps \cite{Bao-2024-towardscompleteness}, we have a set of hypercube vertices and one identifies some subset of vertices as the same, or equivalently this can be thought of as taking a specific partition of the set of vertices. The new edge after vertex identification are derived from the ones existing before identification, i.e., if two vertices are identified then the set of neighboring vertices is the union of their neighboring vertices before the identification. This ensures that the image graph is always a contraction. However, it is not guaranteed that the image graph is isometrically embedabble in a hypercube. In the case when the image graph can be isometrically embedded in a hypercube, we get a partial cube and a valid HEI. Therefore, we can extract at least two independent conditions for the existence of a HEI, the first is that the image graph is a partial cube, and the second is that the edge map respects the adjacency condition $d_H\adj d_G$. 

In the case of inequality \ref{eq:deformed-mmi}, while the image graph is a partial cube, there is no partition of vertex sets preserving edge adjacency that can realize the image graph in figure \ref{fig:mmi-deformed}. In other words, the associated image graph does not come from a contraction of the hypercube $H_3$ but rather it comes from a contraction of a cubical graph\footnote{Cubical graphs are not always isometrically embeddable in a hypercube.} (which is a subgraph of $H_3$). Such cubical graphs, that are not hypercube embeddable, do not correspond to sensible holographic geometries.  More precisely, the non-contractiveness and the breakdown of adjacency manifests as a breakdown of the geodesic structure in the bulk.

Therefore we see that a holographic entropy inequality corresponds to a graph mapping from a hyperplane arrangement of a graph to that of another respecting certain adjacency conditions ensuring the existence of a smooth geometry corresponding to the graph. We will use this insight to prove the completeness of the contraction map proof method.

Before moving on from this example, let us discuss the image graph in terms of the corresponding geometry we chose earlier. It is clear that one can never get the RT surface arrangement of figure \ref{fig:mmi-deformed-surfaces} by cutting and gluing RT surfaces in figure \ref{fig:lhs-terms}.

\begin{figure}[t]
    \centering
    \includegraphics[width=0.4\linewidth]{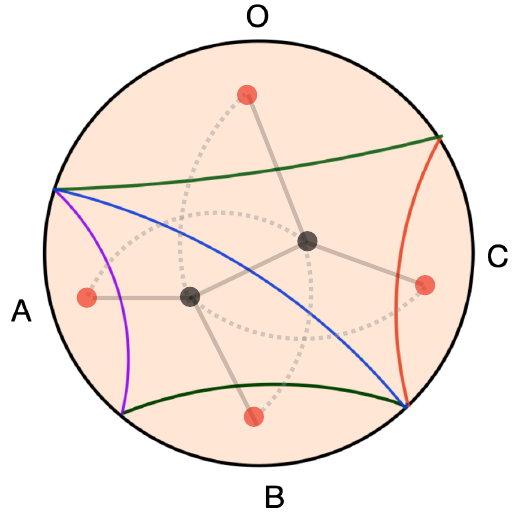}
    \caption{The minimal surface arrangement corresponding to the RHS of inequality \ref{eq:deformed-mmi}. It is impossible to smoothly deform the configuration of RT surfaces in figure \ref{fig:lhs-terms} into this configuration by cutting and gluing. It requires the RT surface associated with $AB$ to be deformed and not deformed simultaneously.}
    \label{fig:mmi-deformed-surfaces}
\end{figure}

\subsection{A Five-Party Example}
The inequality considered in subsection \ref{subsec:mm-deformed} is neither balanced, nor superbalanced. The example may appear a bit artificial and trivial to the reader. Now we will consider a five-party example,
\begin{equation}\label{eq:5partyinvalid}
    S(BC)+S(ACD)+S(ACE)+S(BCDE) \geq S(AC)+S(BCD)+S(BCE)+S(ACDE),
\end{equation}
which can be rephrased as,
\begin{equation}
    I(D : E | AC) - I(D: E| BC)\geq 0
\end{equation}
where we have defined
\begin{equation}
\begin{aligned}
    & I(X:Y):= S(X)+S(Y)-S(XY), \\
    & S(X|Y):= S(XY)-S(Y).
\end{aligned}
\end{equation}
Since the inequality \ref{eq:5partyinvalid} can be expressed as a difference of mutual information between $D$ and $E$ conditioned on $AC$ and $BC$ respectively, one can choose a geometry where the inequality is violated. Therefore, this is not a valid HEI. As expected, one cannot find a contraction map for this inequality. Using the deterministic rules for filling up the contraction map \cite{Bao-2024-properties}, one reaches a contradiction where a single bit is simultaneously required to be $0$ and $1$. We plot the graph in figure \ref{fig:5party-invalid-graph}  associated with the partially filled contraction map and see that the image graph is not isometrically embeddable in $H_4$\footnote{Note that the remaining four bitstrings are partially fixed to be $f(1000)\rightarrow 0--0$, $f(1110) \rightarrow 1--1$, $f(1010),f(1100)\rightarrow 0--1$. No choice for the remaining bits can give a contraction map.}. Therefore, the associated inequality is invalid. It is easy to see that this graph is a partial cube and can be embedded in a higher dimensional hypercube (e.g., in $H_6$). However, this will still yield a false inequality as the graph mapping does not preserve adjacency, i.e., it will yield a mapping where the pre-image does not satisfy the adjacency condition $d_H\adj d_G$.

\begin{figure}[t]
    \centering
    \includegraphics[width=0.7\linewidth]{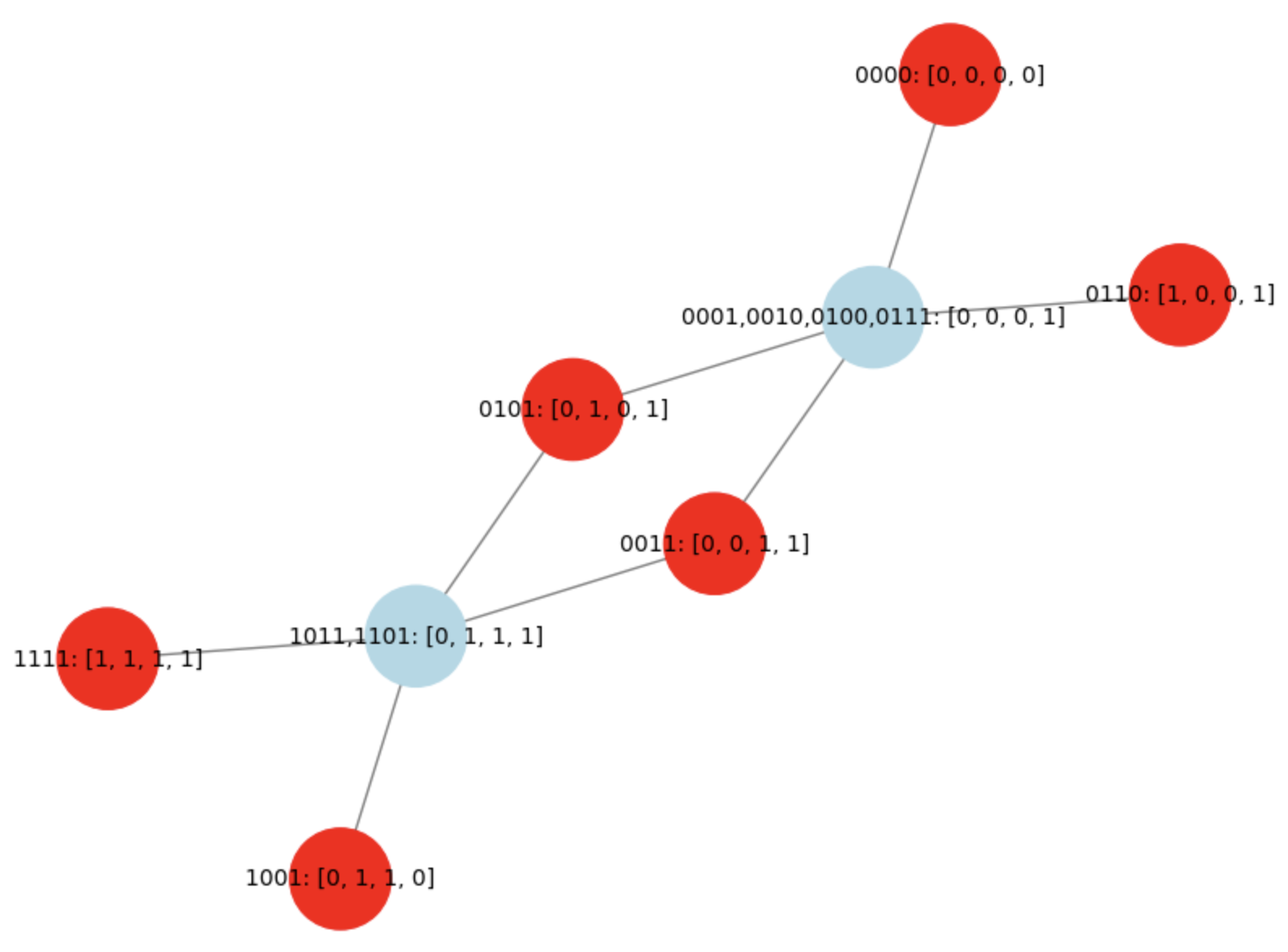}
    \caption{The graph corresponding to the contraction map after choosing the empty bits to be $1,1,0,0$ respectively. This graph (encoding the HEI through occurrence bitstrings) does not have an isometric embedding in $H_5$. The vertices in red correspond to occurrence bitstrings, and those in light-blue are fixed by the deterministic rules described in \cite{Bao-2024-properties}.}
    \label{fig:5party-invalid-graph}
\end{figure}

\subsection{Non-contraction and Geometry}
We will now look at an example of how a non-contraction map induces a geometry that does not satisfy the HEIs.

The choice of the boundary conditions we have in table \ref{tab:non-cmap-ex} gives the MMI \eqref{eq:mmi}, but the map $\tilde{f}$ is non-contractive. For example, the bitstring $001$ of the second row of table \ref{tab:non-cmap-ex} is non-contractive with respect to $000$ and $011$, i.e.,
\begin{equation}
    d_H(001,000) = 1 < d_H(\tilde{f}(001),\tilde{f}(000)) = 2,
\end{equation}
and
\begin{equation}
    d_H(001,011) = 1 < d_H(\tilde{f}(001),\tilde{f}(011)) = 2.
\end{equation}

Before presenting the example, we clarify here that an inequality is false if and only if no contraction map can be found. An inequality read off from a non-contraction map does not say anything about its validity. In this section, we provide how a non-contraction map induces a geometry that does not satisfy the MMI. In section \ref{sec:proof}, we prove that any non-contraction map induces such a geometry. Hence, if one cannot find any contraction map for a given inequality, the inequality is false.

\begin{figure}[t]
    \centering
    \begin{subfigure}{0.3\textwidth}
        \centering
        \includegraphics[width=\linewidth]{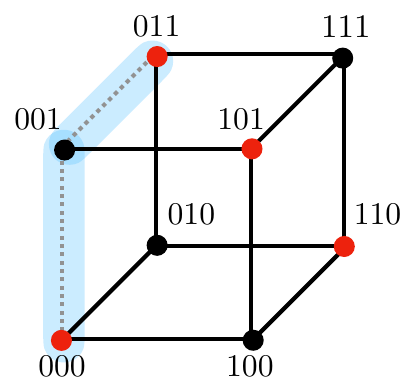}
        \caption{\small{$\mP(\Phi_{\tilde{f}}(H_3))\cup\{\tilde{0}_{\tilde{E}}\}$}}
        \label{fig:mmi_noncontraction_preimage}
    \end{subfigure}\hfill
    \begin{subfigure}{0.4\textwidth}
        \centering
        \includegraphics[width=\linewidth]{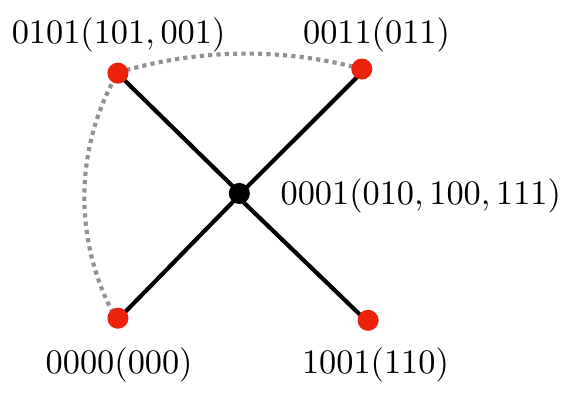}
        \caption{\small{$\Phi_{\tilde{f}}(H_3)$}}
        \label{fig:mmi_noncontraction_image}
    \end{subfigure}\hfill
    \caption{\small{The dotted lines stress the null edges between the vertices that do not satisfy the adjacency condition. (a) The nontrivial preimage $\mP(\Phi_{\tilde{f}}(H_3))\cup\{\tilde{0}_{\tilde{E}}\}$ whose two null edges are replaced by the virtual edges highlighted blue. (b) the image $\Phi_{\tilde{f}}(H_3)$. }}
    \label{fig:mmi-noncontraction-graph}
\end{figure}

Consider a choice of a connected subset of bitstrings, and its image under the non-contraction map $\tilde{f}$, 
\begin{equation}
    X:=\{000,001,011,101,110,11\},\; \tilde{f}(X) = \{0000,0001,0011,0101,1001\}.
\end{equation}
With the isometric condition $d_H=d_G$ with respect to the hypercubes $H_{3}$ and $H_{4}$, we have the path isometries, $\iota_M$, $\iota_N$, and the graph map $\Phi_{\tilde{f}}$. 

Using the path isometries, we have the unit-weighted graphs $G$ and $\Phi_{\tilde{f}}(G) \sim \Phi_{\tilde{f}}(H_{3})$\footnote{One can easily check that $\Phi_{\tilde{f}}(G)$ is graph isomorphic to $\Phi_{\tilde{f}}(H_{3})$}, which satisfy the adjacency condition $d_H\adj d_G$. Then, we can identify the nontrivial preimage $\mP(\Phi_{\tilde{f}}(H_{3}))$ and the image $\Phi_{\tilde{f}}(H_{3})$ as in figure \ref{fig:mmi-noncontraction-graph}. 

The nontrivial preimage $\mP(\Phi_{\tilde{f}}(H_3))$ is a cubical graph where the vertices $000,001,011$ do not satisfy the adjacency condition $d_H\adj d_G$. As a result, the edges $(001,011)$ and $(001,000)$ are absent. This is simply due to the fact that the kernel $Ker(\Phi_{\tilde{f}})$ of the graph contraction map $\Phi_{\tilde{f}}$ is non-trivial.

It should be noticed that there is another graph $\tilde{G}=(\tilde{V}_X,\tilde{E}_X)$ whose image graph $\Phi_{\tilde{f}}(\tilde{G})$ under the map $\Phi_{\tilde{f}}$ is also graph isomorphic to $\Phi_{\tilde{f}}(H_{3})$. For the map $\tilde{f}$ to give the MMI, any subgraphs of the domain graph that can be mapped to a connected subgraph with a set of proper boundary vertices of the image graph $\Phi_{\tilde{f}}$ should satisfy the inequality as discussed in section \ref{sec:review-cmap}. Hence, in our case, the subgraphs $G$ and $\tilde{G}$ with the graph non-contraction map should satisfy the inequality for the non-contraction map to give the MMI.

Thus, to verify the inequality, we consider the quantity,
\begin{equation}\label{eq:hei-verifier}
    \sum_{x,x'\in X}\Big( d_H(x,x') - d_H(\tilde{f}(x) ,\tilde{f}(x')) \Big) |(\iota_M(x),\iota_M(x'))|,
\end{equation}
and check its sign. As in \eqref{eq:hei-verifier-positive}, its sign is positive regardless of the choice of the edge weights if the map $\tilde{f}$ is contractive.  When the map $\tilde{f}$ is non-contractive, for instance, $d_H(\tilde{x},\tilde{x}') + \kappa = d_H(\tilde{f}(\tilde{x}),\tilde{f}(\tilde{x}'))$ for $\kappa \in \mbZ_{+}$ and a pair $\tilde{x},\tilde{x}'\in X$, then, \eqref{eq:hei-verifier} is not necessarily positive because of the pair $\tilde{x},\tilde{x}'$, i.e.,
\begin{equation}\label{eq:hei-verifier-noncontraction}
    \sum_{x,x'\in X}\Big( d_H(x,x') - d_H(\tilde{f}(x) ,\tilde{f}(x')) \Big) |(\iota_M(x),\iota_M(x'))| -\kappa  |(\iota_M(\tilde{x}),\iota_M(\tilde{x}'))|.
\end{equation}

We now see that \eqref{eq:hei-verifier} can be negative for $\tilde{G}$. The key is the edge weights of the edges, such as $(001,000)$ and $(001,011)$, which violate the adjacency condition. We call these edges \textit{virtual edges} and are defined in definition \ref{def:virtual-edges}. Without loss of generality, their edge weights can be chosen to be infinite. 

Physically, this means that the non-contraction map destroys the smoothness of the bulk geometry. In other words, there are no geodesics between any spatial points in the bulk subregion associated with $001$ and any spatial points in that associated with, for instance, $000$, see figure \ref{fig:mmi-noncontraction-surfaces}.

Then, the edge weights of the virtual edges give an infinite value to the second term in \eqref{eq:hei-verifier-noncontraction}. This results in the sign of \eqref{eq:hei-verifier-noncontraction} being negative, and thus violates the MMI. One should note that this does not mean that the MMI is a false inequality because there exists a contraction map.

\begin{table}[t]
\centering
\begin{tabular}{@{}lccclcccc@{}}
\hline
& \textbf{AB} & \textbf{AC} & \textbf{BC} &  & \textbf{A} & \textbf{B} & \textbf{C} & \textbf{ABC} \\ \cmidrule(r){1-4} \cmidrule(l){6-9} \textbf{O} & \textbf{0} & \textbf{0} & \textbf{0} &  & \textbf{0} & \textbf{0} & \textbf{0} & \textbf{0} \\
& 0 & 0 & 1 &  & 0 & 1 & 0 & 1 \\
& 0 & 1 & 0 &  & 0 & 0 & 0 & 1 \\
\textbf{C} & \textbf{0} & \textbf{1} & \textbf{1} &  & \textbf{0} & \textbf{0} & \textbf{1} & \textbf{1}\\
 & 1 & 0 & 0 &  & 0 & 0 & 0 & 1 \\
\textbf{B} & \textbf{1} & \textbf{0} & \textbf{1} &  & \textbf{0} & \textbf{1} & \textbf{0} & \textbf{1}  \\
\textbf{A} & \textbf{1} & \textbf{1} & \textbf{0} &  & \textbf{1} & \textbf{0} & \textbf{0} & \textbf{1}\\
 & 1 & 1 & 1 & \multirow{-9}{*}{\textbf{}} & 0 & 0 & 0 & 1 \\ \bottomrule
\end{tabular}
\caption{An example of non-contraction map $f: H_3 \rightarrow H_4$ whose boundary conditions give the MMI.}
\label{tab:non-cmap-ex}
\end{table}

In the next section, we will see that there is always some holographic geometry where the geodesic structure is altered (and, smoothness of the bulk manifold breaks down) when the corresponding map is a non-contraction map. Thus, for inequality candidates which do not have any contraction map, some holographic geometry will always violate it, and thus, the inequality must be false. Therefore, the existence of a contraction map is necessary for the proof of a valid HEI. We will formalize this statement in the following section \ref{sec:proof}.

\begin{figure}[t!]
    \centering
    \begin{subfigure}{0.3\textwidth}
        \centering
        \includegraphics[width=\linewidth]{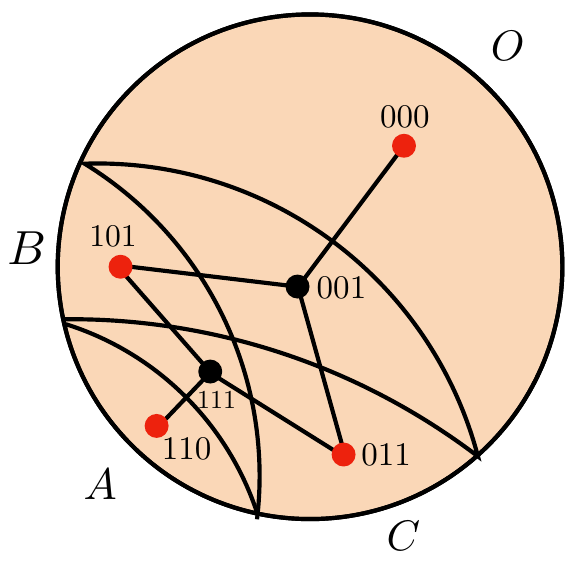}
        \subcaption{$G$}
        \label{fig:mmi-noncontraction-G}
    \end{subfigure}\hfill
    \begin{subfigure}{0.3\textwidth}
        \centering
        \includegraphics[width=\linewidth]{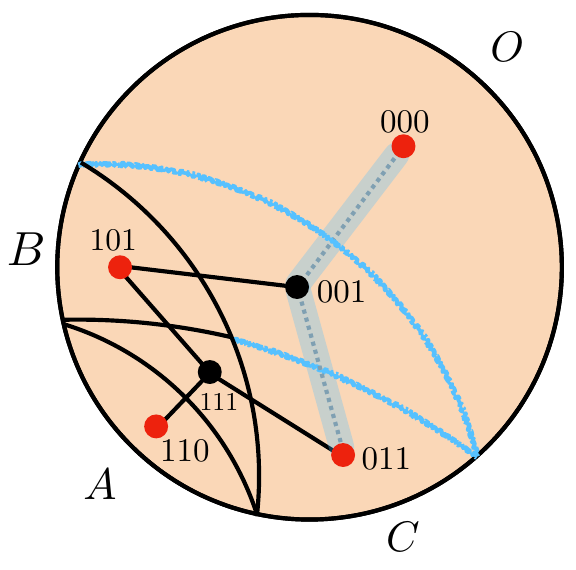}
        \label{fig:mmi-noncontraction-G}
        \subcaption{$\tilde{G}$}
    \end{subfigure}\hfill
    \caption{\small{The RT arrangements on AdS$_3$/CFT$_2$ corresponding to the subgraphs $G,\tilde{G}$: (a) The edge weights of all the edges are associated with the areas of the RT surfaces. (b) The edge weights of the edges colored black are associated with the areas of the corresponding RT surfaces, while the edge weights of the virtual edges highlighted blue are large or infinite. The two blue lines on the constant time slice represent the obstruction to the smoothness of the manifold.}}
    \label{fig:mmi-noncontraction-surfaces}
\end{figure}

\section{Main theorem} \label{sec:proof}

\begin{theorem}[Completeness of contraction map; HEI $\leftrightarrow$ contraction map]\label{thm:completeness}
    Given a candidate entropy inequality,
    \begin{equation}\label{eq:HEI}
        \sum_{i=1}^l \alpha_i S_{P_i} \geq \sum_{j=1}^r \beta_j S_{Q_j}, \; \alpha_i,\beta_j \in \mbZ_{+},
    \end{equation}
    it is a valid HEI if and only if a contraction map $f:\{0,1\}^M\to \{0,1\}^N$  exists (where $M=\sum_i \alpha_i$ and $N=\sum_j \beta_j$) satisfying homology conditions on the boundary\footnote{\label{ftn:rescaling_main_thm}A contraction map may depend on the particular presentation of a holographic entropy inequality. Nevertheless, modifying the presentation by merging or expanding terms with addition, such as \eqref{eq:genentineq-expand} or uniformly rescaling all coefficients, such as $c\sum_{i=1}^l \alpha_i S_{P_i} \geq  c \sum_{j=1}^r \beta_j S_{Q_j}$ for $\alpha_i,\beta_j,c \in \mbZ_{+}$, may yield a different contraction map that still proves the same inequality, because these operations do not alter its physical content. In this paper, we provide a proof with HEIs with unimodular coefficients. However, the proof extends straightforwardly to inequalities obtained by uniformly rescaling all coefficients. An alternative proof of the completeness is provided in  \cite{GHHS:preparation}.}.
    
\end{theorem}

From theorem \ref{thm:proofbycontraction} \cite{Bao:2015bfa}, the HEI \ref{eq:HEI} is valid if there is a corresponding contraction map. The main contribution of this paper is to prove the converse, i.e., theorem \ref{thm:HEI-to-contraction-map}.

\begin{theorem}[HEI$\rightarrow$contraction map]\label{thm:HEI-to-contraction-map}
    If the candidate HEI \ref{eq:HEI} is valid, there exists a contraction map.
\end{theorem}

Recall that in the holographic dictionary, the RT formula \cite{Ryu-2006-RTformula} relates the leading order entanglement entropy of a holographic quantum state on the boundary CFT to a minimal surface (RT surface) in the bulk. Therefore, a valid HEI holds for any holographic geometry obeying the RT formula. We take this as the definition of a HEI.

\begin{definition}\label{def:HEI}
    An entropy inequality is a valid HEI if and only if every holographic geometry obeying the RT formula satisfies the inequality.
\end{definition}

Theorem \ref{thm:HEI-to-contraction-map} is proved by contradiction. We assume, on the contrary, that there exists a valid HEI that does not have any corresponding contraction map. We first prove in lemma \ref{lem:nonhologrpahicgeometry} that, for any non-contraction map, there is at least one geometry that does not satisfy the inequality. Therefore, from definition \ref{def:HEI}, the inequality cannot be a valid HEI. Hence, theorem \ref{thm:HEI-to-contraction-map} is proved. Taking theorem \ref{thm:proofbycontraction} and theorem \ref{thm:HEI-to-contraction-map} together, we prove theorem \ref{thm:completeness}.

\subsection{Mathematical proof}

Given a convex combination of entanglement entropies in the inequality \eqref{eq:HEI},
\begin{equation}\label{eq:HEs}
    S_{[l]}: =\sum_{i=1}^l \alpha_i S_{P_i}, \; S_{[r]}: =\sum_{j=1}^r \beta_j  S_{Q_j}, \; \alpha_i,\beta_j \in \mbZ_{+},
\end{equation}
or
\begin{equation}
    S_M: =\sum_{u=1}^M  S_{L_u}, \; S_N: =\sum_{v=1}^N  S_{R_v}
\end{equation}
for the unimodular sums of \eqref{eq:HEs}.
We denote the corresponding discrete entropies on a graph by $S^*_{[l]},\;S^*_{[r]}$ and $S^*_M, \; S^*_N$, respectively. 

Let us now define a non-contraction map.
\begin{definition}[Non-contraction map]\label{def:non-contractionmap}
    A map $\tilde{f}: \{0,1\}^M\to \{0,1\}^N$ is a non-contraction map if, for some $x,x'\in \{0,1\}^M$,
    \begin{equation}\label{eq:non-contraction}
        d_H(x,x') < d_H(\tilde{f}(x),\tilde{f}(x')) .
    \end{equation}
\end{definition}

We assume the following.
\begin{assumption}\label{amp:non-contractionmap}
    There exists a valid HEI such that there is no corresponding contraction map\footnote{As mentioned in footnote \ref{ftn:rescaling_main_thm}, the non-existence of a contraction map means that none exists for any presentation of an inequality obtained by merging or expanding terms of the inequality and uniformly rescaling its coefficients.}.     
\end{assumption}
For a map to give a HEI, any subgraph of $H_M$, or equivalently, any RT arrangement on a holographic geometry, induced by the map $\tilde{f}$ satisfies the given HEI. This statement can be unraveled into two folds, similarly discussed in section \ref{sec:review-cmap}.

First, the choice of a map $\tilde{f}:\{0,1\}^M\to \{0,1\}^N$ fixes the image $Im(\tilde{f})\subseteq\{0,1\}^N$. The isometric condition $d_H=d_G$ with respect to the $J$-dimensional hypercubes $H_{J}=(V_{J},E_{J})$ for $J=M,N$ in \eqref{eq:dH_dG_notation} induces two objects: i) path isometries, $\iota_M:\{0,1\}^M \to V_M$, $\iota_N:\{0,1\}^N \to V_N$, and ii) the graph map $\Phi_{\tilde{f}}:H_M\to H_N$ associated with $\tilde{f}$. Similarly to the map $\tilde{f}$, its image graph $\Phi_{\tilde{f}}(H_M)\subseteq H_N$ is fixed by the map $\tilde{f}$. 

Second, if the map $\tilde{f}$ can give an HEI, any subgraph $G_M\subseteq H_M$ with a set of boundary vertices, which can be mapped to a connected subgraph $G_N \subseteq \Phi_{\tilde{f}}(H_M)$, should satisfy the inequality\footnote{This is due to the geometry-graph duality\cite{Bao:2015bfa}.}. The subgraphs $G_M$ are identified as the subgraphs of the non-trivial preimage $\mP(\Phi_{\tilde{f}}(H_M))\subseteq H_M$ of $\Phi_{\tilde{f}}(H_M)\subseteq H_N$. In this sense, we say that the subgraphs $G_M\subseteq \mP(\Phi_{\tilde{f}}(H_M)) $, or the corresponding RT arrangements on holographic geometries, are induced by the map $\tilde{f}$.

The strategy of the proof is to show that there exists at least one geometry with an RT arrangement induced by a non-contraction map $\tilde{f}$ that does not satisfy the inequality. This is true for any choice of non-contraction maps satisfying the homology conditions on the boundary of the HEI. Hence, it contradicts assumption \ref{amp:non-contractionmap}. Thus, theorem \ref{thm:HEI-to-contraction-map} is proved.

Now, we begin the proof. From the isometric condition $d_H=d_G$ with respect to $J$-dimensional hypercubes $H_{J}=(V_{J}, E_{J})$ for $J=M,N$, we have the path isometries $\iota_M:\{0,1\}^M\to V_M$ and $\iota_N:\{0,1\}^N\to V_N$ such that
\begin{equation}
\begin{split}
    d_H(x,x') &= d_G(\iota_M(x),\iota_M(x')), \; \forall x,x' \in\{0,1\}^M\\
    d_H(y,y') &= d_G(\iota_N(y),\iota_N(y')), \; \forall y,y' \in\{0,1\}^N.\\
\end{split}
\end{equation}
Then, we can construct a graph map from the non-contraction map $\tilde{f}$ with the isometries defined above, i.e., the graph map $\Phi_{\tilde{f}}:=(\tilde{\phi}^V,\tilde{\phi}^E)$ consists of a vertex map $\tilde{\phi}^V$ and an edge map $\tilde{\phi}^E$ defined as
\begin{equation}
    \tilde{\phi}^V\circ \iota_M := \iota_N \circ \tilde{f},\; \tilde{\phi}^E(\cdot,\cdot):= (\tilde{\phi}^V\circ \iota_M(\cdot),\tilde{\phi}^V\circ \iota_M(\cdot)).
\end{equation}

For $\tilde{x},\tilde{x}'\in\{0,1\}^M$ satisfying the non-contraction condition \eqref{eq:non-contraction},
\begin{equation}
    d_H(\tilde{x},\tilde{x}') <d_H(\tilde{f}(\tilde{x}),\tilde{f}(\tilde{x}')),
\end{equation}
we have
\begin{equation}\label{eq:graph-noncontraction}
    d_G(\iota_M(\tilde{x}),\iota_M(\tilde{x}'))<d_G(\tilde{\phi}^V \circ \iota_M(\tilde{x}),\tilde{\phi}^V\circ \iota_M(\tilde{x}')) .
\end{equation}

Now, we state the properties of the graph non-contraction maps.
\begin{lemma}\label{lem:kernel-preimage}
    For a graph non-contraction map $\Phi_{\tilde{f}}:H_M \to H_N$, its kernel $Ker(\Phi_{\tilde{f}})$ is non-trivial. Then, $\mP(\Phi_{\tilde{f}}(H_M))$ is a proper subgraph of the hypercube $H_M$, i.e., $\mP(\Phi_{\tilde{f}}(H_M)) \subset H_M$. Equivalently, $\mP(\Phi_{\tilde{f}}(H_M))$ does not satisfy $d_H\adj d_G$.
\end{lemma}
\begin{proof}

    The kernel of the vertex map is trivial by definition \ref{def:non-contractionmap}.

    Consider a non-contraction map $\tilde{f}:\{0,1\}^M\to\{0,1\}^N$ for $M,N\in \mbZ$. Then, by definition \ref{def:non-contractionmap}, there exists an edge $\tilde{e}=(\iota_M(\tilde{x}), \iota_M(\tilde{x}')) \in E_M$ such that
    \begin{equation}
        d_H(\tilde{x},\tilde{x}') <d_H(\tilde{f}(\tilde{x}),\tilde{f}(\tilde{x}')).
    \end{equation}
    As in \eqref{eq:graph-noncontraction}, we have
    \begin{equation}
        d_G(\iota_M(x), \iota_M(x')) =1 < d_G(\iota_N\circ \tilde{f}(x), \iota_N\circ \tilde{f}(x')) =  d_G( \tilde{\phi}^V\circ \iota_M(x), \tilde{\phi}^V\circ \iota_M(x')) .
    \end{equation}
    This implies that the edge $\tilde{\phi}^E(\tilde{e})=( \tilde{\phi}^V\circ \iota_M(x), \tilde{\phi}^V\circ \iota_M(x')) $ is a null edge, i.e.,
    \begin{equation}
        \tilde{\phi}^E(\tilde{e})=\emptyset_E \in E_N.
    \end{equation}
    Thus, $\tilde{e} \in Ker(\tilde{\phi}^E)$. We denote any edge in $Ker(\tilde{\phi}^E)$ by $\tilde{e}$ from now on. 

    From definition \ref{def:preimage-graph} and \ref{def:kernel-graph}, the vertex set $\tilde{V}_M$ and the edge set $\tilde{E}_M$ of $\mP(\Phi_{\tilde{f}}(H_M))=(\tilde{V}_M,\tilde{E}_M)$ are
    \begin{equation}\label{eq:tilde-vertex-edge-set}
    \begin{split}
        \tilde{V}_M = V_M \setminus Ker(\tilde{\phi}^V),\\
        \tilde{E}_M = E_M \setminus Ker(\tilde{\phi}^E).
    \end{split}
    \end{equation}
    In our case, $\tilde{V}_M  = V_M$ because $Ker(\tilde{\phi}^V)=\{\emptyset_V\}$. With the non-trivial $Ker(\tilde{\phi}^E)$, we have 
    \begin{equation}
        \mP(\Phi_{\tilde{f}}(H_M))= H_M \setminus Ker(\Phi_{\tilde{f}})=(V_M,\tilde{E}_M) \subset H_M.
    \end{equation}
    It is straightforward to see that $\mP(\Phi_{\tilde{f}}(H_M))$ violates the adjacency condition $d_H\adj d_G$.

\end{proof}

For later purposes, we will now define \textit{virtual edges}. 
\begin{definition}[Virtual edges]\label{def:virtual-edges}
    Consider a graph $G=(V,E)$. We define a virtual edge $0_E$ between non-adjacent vertices $v,v'\in V$. That is,
    \begin{equation}\label{eq:virtual-edge}
        0_E: = (v,v') 
    \end{equation}
    for $v,v'\in V$ such that
    \begin{equation}
        d_G(v,v') >1.
    \end{equation}
    We set its edge weight $|0_E|$ infinite \footnote{In principle, the edge weight $|0_E|$ can be set to an arbitrarily large but finite value so that it is larger than the largest distance in $G$. For a weighted graph $G$, this means that $|0_E|$ is larger than its largest distance, i.e., the largest sum of edge weights over all paths. Without loss of generality, we set it to be infinite.}, i.e.,
    \begin{equation}\label{eq:virtual-edgeweight}
        |0_E|:=\infty.
    \end{equation}
\end{definition}

We can replace the null edges of a graph between non-adjacent vertices with virtual edges without changing the graph structures\footnote{Graph distance between any pairs of vertices is preserved.}. We assume from now on that the null edges of any graph appearing below is replaced by the virtual edges.

Especially, in our case, there are virtual edges in $\tilde{E}_M$ of $\mP(\Phi_{\tilde{f}}(H_M))$ that used to be the edges $\tilde{e} \in Ker(\tilde{\phi}^E)\subset E_M$ of $H_M$ (see figure \ref{fig:mmi-noncontraction-graph}). We denote those as $\tilde{0}_{\tilde{E}}$. To stress this point, we write the replacement by
\begin{equation}\label{eq:tilde-edge-virtual-edges}
     \tilde{E}_M = E_M \setminus Ker(\tilde{\phi}^E) \cup \{\tilde{0}_{\tilde{E}}\}.
\end{equation}

\begin{lemma}\label{lem:failing-inequality}
    Consider a non-contraction map $\tilde{f}$ satisfying the homology conditions for an inequality of discrete entropy inequality
    \begin{equation}
        S^*_M \geq S^*_N.
    \end{equation}
    Then, there is a unit-weighted subgraph $\tilde{G}_X\subseteq\mP(\Phi_{\tilde{f}}(H_M))$ that does not satisfy $d_H\adj d_G$ and gives 
    \begin{equation}
        S^*_M <  S^*_N
    \end{equation}
    with any choice of finite edge weights.
\end{lemma}
\begin{proof}
    $\tilde{G}_X:=\{\tilde{V}_X, \tilde{E}_X\}\subseteq\mP(\Phi_{\tilde{f}}(H_M))$ can be defined as 
    \begin{equation}
        \tilde{V}_X:= \{\tilde{\iota}_M(x)| \forall x \in X_M\},\; \tilde{E}_X:=\{(\tilde{\iota}_M(x),\tilde{\iota}_M(x'))|x,x'\in X_M\},
    \end{equation}
    where $X_M \in \{0,1\}^M$. Since $\tilde{G}_X$ does not satisfy $d_H\adj d_G$, we have a map $\tilde{\iota}_M:\{0,1\}^M \to \tilde{V}_M$ such that there is at least one virtual edge $\tilde{0}_{\tilde{E}}= (\tilde{\iota}_M(\tilde{x}),\tilde{\iota}_M(\tilde{x}') )$ for $\tilde{x}, \tilde{x}'\in X_M$ satisfying
    \begin{equation}\label{eq:HneqG-edge}
        d_H(\tilde{x},\tilde{x}') < d_G(\tilde{\iota}_M(\tilde{x}),\tilde{\iota}_M(\tilde{x}')).
    \end{equation}
    Without loss of generality, suppose there is a virtual edge $\tilde{0}_{\tilde{E}} = (\iota_M(\tilde{x}),\iota_M(\tilde{x}'))$ of $\tilde{G}_X$ such that 
    \begin{equation}\label{eq:k-non-contractive}
        d_H(\tilde{x},\tilde{x}') + \kappa =d_H(\tilde{f}(\tilde{x}),\tilde{f}(\tilde{x}'))
    \end{equation}
    for some positive integer $\kappa>0$. 

    After choosing edge weights on $\tilde{G}_X$, the inequality for $\tilde{G}_X=(\tilde{V}_X,\tilde{E}_X)$ can be computed by 
    \begin{equation}
    \begin{split}
        S^*_M-\mbS^*_N 
        =&\sum_{(\iota_M(x),\iota_M(x'))\in \tilde{E}_X}\Big( d_H(x,x') - d_H(\tilde{f}(x) ,\tilde{f}(x')) \Big) |(\iota_M(x),\iota_M(x'))|\\
        & \qquad \qquad \qquad +\Big( d_H(\tilde{x},\tilde{x}') - d_H(\tilde{f}(\tilde{x}) ,\tilde{f}(\tilde{x}')) \Big)|\tilde{0}_{\tilde{E}}|\\
        =&\sum_{(\iota_M(x),\iota_M(x'))\in \tilde{E}_X}\Big( d_H(x,x') - d_H(\tilde{f}(x) ,\tilde{f}(x')) \Big) |(\iota_M(x),\iota_M(x'))|\\
        &\qquad \qquad \qquad -\kappa |\tilde{0}_{\tilde{E}}|.
    \end{split}
    \end{equation}
    We used \eqref{eq:k-non-contractive}, i.e., $d_H(\tilde{x},\tilde{x}') -d_H(\tilde{f}(\tilde{x}),\tilde{f}(\tilde{x}'))= - \kappa $, to obtain the coefficient $\kappa$ of the second term. The first sum is positive semi-definite and finite. However, the second term is infinite. Hence, we have 
    \begin{equation}
        S^*_M-\mbS^*_N = -\infty <0.
    \end{equation}
    
    The optimization of the edge weights cannot remove the presence of the virtual edge. Therefore, the graph $\tilde{G}_X$ cannot give the inequality, i.e.,
    \begin{equation}
        S^*_M-S^*_N= -\infty < 0.
    \end{equation}
    
\end{proof}

\begin{lemma}\label{lem:nonhologrpahicgeometry}
    Consider a candidate HEI. For a non-contraction map $\tilde{f}$ satisfying the homology conditions on the boundary of the HEI, there exists at least one geometry that does not satisfy the HEI.
\end{lemma}
\begin{proof}

    Consider $X_M\subseteq \{0,1\}^M$ and a path isometries $\iota_X:\{0,1\}^M\to V_M$ which induces $G_X\subset H_M$ and a map $\tilde{\iota}_X$ which induces  $\tilde{G}_X\subset \mP(\Phi_{\tilde{f}}(H_M))$ violating the adjacency condition $d_H\adj d_G$. In particular,
    \begin{equation}\label{eq:GtildeG}
    \begin{split}
        G_X &= (V_X,E_X)\subseteq H_M\\
        \tilde{G}_X &=(\tilde{V}_X,\tilde{E}_X)\subseteq \mP(\Phi_{\tilde{f}}(H_M))\subset H_M\\
    \end{split}
    \end{equation}
    where
    \begin{equation}
        V_X:= \{\iota_M(x)|x\in X_M\},\; \tilde{V}_X:= \{\tilde{\iota}_M(x)|x\in X_M\}
    \end{equation}
    and 
    \begin{equation}
        E_X:= \{(\iota_M(x),\iota_M(x'))| x,x'\in X_M\},\; 
        \tilde{E}_X:= \{(\tilde{\iota}_M(x),\tilde{\iota}_M(x'))| x,x'\in X_M\}\\
    \end{equation}
    Note that $G_X \nsubseteq  \mP(\Phi_{\tilde{f}}(H_M))$. Moreover, one can write
    \begin{equation}
        \tilde{V}_X = V_X,\;\tilde{E}_X = E_X\setminus Ker(\tilde{\phi}^E)\cup\{\tilde{0}_{\tilde{E}}\}.
    \end{equation}
    as in \eqref{eq:tilde-vertex-edge-set} and \eqref{eq:tilde-edge-virtual-edges}.

    The edges $\tilde{e}\in Ker(\tilde{\phi}^E)$ in $G_X$ is replaced with $\tilde{0}_{\tilde{E}}$ in $\tilde{G}_X$, see figure \ref{fig:mmi-noncontraction-graph}. These are the virtual edges that satisfy
    \begin{equation}\label{eq:HneqG-edge}
        d_H(\tilde{x},\tilde{x}') =1 < d_G(\tilde{\iota}_M(\tilde{x}),\tilde{\iota}_M(\tilde{x}')).
    \end{equation}

    To compute the discrete entropy $S_{L_u}^*$ for $G_X$ and $\tilde{G}_X$, we elucidate two properties of the graph, i.e., the Winkler classes in definition \ref{def:winkler-relation-class} and the edge weights. 
    
First, the homology conditions and the bitstrings in $X$ determine the Winkler classes $\Theta_u[e]$, or $\Theta_u$, for $u=1,\cdots,M$ of the graphs. This requires that $\tilde{e} = (\iota_M(\tilde{x}),\iota_M(\tilde{x}')) \in \Theta_u$ for $G$ and $\tilde{0}_{\tilde{E}} =(\tilde{\iota}_M(\tilde{x}),\tilde{\iota}_M(\tilde{x}')) \in \Theta_u$ for $\tilde{G}_X$.    
Second, the weights of the edges (except the virtual edges) in the graphs $G_X$ and $\tilde{G}_X$ have been chosen to be unity, so far. Now we will introduce the edge weights associated with the RT arrangement to these graphs. Let us assume that the edges of $G_X$ and those of $\tilde{G}_X$ have the same edge weights except $\tilde{e}$ and $\tilde{0}_{\tilde{E}}$. The choice of the edge weights fixes the RT arrangement. In particular, the edge weight of the edges in each class $\Theta_u$ is fixed by the choice and, thus, so is the area of the RT surface corresponding to the class.

Then, with the choice of edge weights $|e|,|\tilde{e}|$ for $e,\tilde{e}\in E_X$ for $G_X$, the discrete entropy $S_{L_u}^*$ of $G_X$ is given by 
\begin{equation}
    S_{L_u}^* = \sum_{e \in \Theta_u} |e| + \sum_{\tilde{e} \in \Theta_u} |\tilde{e}|,
\end{equation}
which is positive semi-definite and can be tuned to be finite.
However, the discrete entropy $S_{L_u}^*$ for $\tilde{G}_X$ is divergent by definition \eqref{eq:virtual-edgeweight}, i.e.
\begin{equation}
    S_{L_u}^* = \sum_{e \in \Theta_u} |e| + \sum_{\tilde{0}_{\tilde{E}} \in \Theta_u} |\tilde{0}_{\tilde{E}}| = \infty.
\end{equation}
Hence, any geometry realized by $\tilde{G}_X$ cannot give the proper holographic entanglement entropy.

From lemma \ref{lem:failing-inequality}, we see that the graph $\tilde{G}_X$ cannot satisfy the inequality, i.e.,
\begin{equation}
    S^*_M-S^*_N< 0.
\end{equation}

In short, for any non-contraction map, there is at least one subgraph $\tilde{G}_X$ in the preimage, which does not satisfy the adjacency condition $d_H\adj d_G$ and the given holographic (discrete) entropy inequality. This implies that any geometry induced from the subgraph $\tilde{G}_X$ with a choice of its edge weights on the equivalence classes $\Theta_u$ cannot satisfy the HEI. Therefore, lemma \ref{lem:nonhologrpahicgeometry} is proved.

\end{proof}

From lemma \ref{lem:nonhologrpahicgeometry} and definition \ref{def:HEI}, any non-contraction map cannot give the HEI. Hence, theorem \ref{thm:HEI-to-contraction-map} is proved. Theorem \ref{thm:HEI-to-contraction-map} and theorem \ref{thm:proofbycontraction}
prove the completeness of the contraction map, i.e., theorem \ref{thm:completeness}.

\section{Discussion}
\label{sec:discussion}

\subsection{Completeness of the algorithmic method to generate all the HEIs}

In \cite{Bao-2024-towardscompleteness}, we have proposed an algorithm that can generate HEIs by generating all possible partial cubes under graph contraction maps. By construction, it respects the adjacency condition $d_H\adj d_G$. In that work, we conjectured that this method can generate all the HEIs. In the current work, we prove this result in theorem \ref{thm:completeness}. We have not, however, proved the optimality of the algorithm in \cite{Bao-2024-towardscompleteness}, and further optimizations upon the algorithm are certainly available. Furthermore, we have not completed the full classifications of HEIs, only giving a constructive method for generating them. We will leave the optimization of those methods and this full classification for future work.

Geometrically, it could be interesting to understand what conditions on the geometry must be relaxed in order for non-contracting inequalities to hold. For example, if a GHZ state were desired as the entanglement pattern between multiple CFTs such as that proposed in \cite{susskind2014ereprghzconsistencyquantum}, what manifold conditions would have to be relaxed for this to be permitted?

\subsection{Interpretations of the HEIs}

It largely remains opaque how to interpret the holographic entanglement entropies heretofore derived; what consequences do they have for holography, or for the subset of quantum states that obey them more broadly? It is known that the cyclic family of inequalities lower bounds the multipartite entanglement of purification and therefore has an operational interpretation \cite{bao2019conditional}, but the operational interpretation of the other ones remain opaque. It is possible that they can serve as finer-grained diagnostics for chaos as in \cite{Hosur_2016} in the context of the tripartite information, or as generalizations of the topological entanglement entropy as investigated in \cite{Naskar:2024mzi}, but neither approach has yet reached full fruition.

It also remains possible that all holographic entanglement inequalities with either fixed party number or fixed number of LHS terms are derived from an inequality with larger numbers of either quantity, where the entropic contributions of specific subsets of parties have been set to zero. This would certainly be conceptually pleasing, as it would lend at least an alternate mathematical unification of the holographic entanglement entropy inequalities.

\subsection{Entanglement wedge nesting and geometry-graph duality}

The geometry-graph duality was proposed and proved in \cite{Bao:2015bfa}. They studied two types of graphs, namely trivalent graphs and partial cubes, in the duality. We utilized partial cubes in this work. They are distinct from the trivalent graphs used in \cite{Bao:2015bfa,HernandezCuenca:2019wgh,AVIS202316,Hernandez-Cuenca:2023iqh,He:2019repackaged,Hernandez-Cuenca:2022marginal} in the sense that the vertices of the trivalent graphs are not labeled with the bitstrings of a contraction map. In contrast, the partial cubes are constructed from the contraction maps as discussed in \cite{Bao-2024-towardscompleteness}. Since the bitstrings of the contraction maps are expected to follow the inclusion/exclusion principle based on the entanglement wedge nesting(EWN) relations, one would expect the EWN relations in the RT arrangements induced from the partial cubes. However, there are bitstrings that do not seem to satisfy the EWN relations and have been known as \textit{unphysical bitstrings} \cite{Avis:2021xnz,Bao-2024-properties}. Such vertices do not enclose any non-zero volume in the bulk.

Recently, \cite{Czech:2025tds} studied and explored the possibility of how the EWN relations constrain the contraction maps. There are three types of EWN violation: The EWN violations i) within a set of bitstrings associated with the LHS of an inequality \cite{Avis:2021xnz,Bao-2024-properties}, ii) that of the RHS of the inequality, and iii) between the set of bitstrings of the LHS and that of the RHS. In particular, the authors in \cite{Czech:2025tds} mainly considered the effect of the third type in the contraction map method. In our case, the third type is not directly relevant to our proof because the key is whether a map is contractive or non-contractive, or equivalently, whether it satisfies the adjacency condition. In other words, it does not matter how many EWN violations occurred.

However, it would still be interesting to explore the properties of a graph contraction map $\Phi_f$ involved with the violations of the third type. This could help us develop an efficient algorithm to find tight HEIs and demystify interpretations of holographic entanglement entropies. We leave further investigations for the future.

\section*{Acknowledgments}
We would like to thank Sergio Hernandez-Cuenca, Bart\l omiej Czech, Hirosi Ooguri, and Michael Walter for helpful comments. N.B. is supported by Northeastern University Department of Physics and by Brookhaven National Laboratory, as well as by the U.S Department of Energy ASCR EXPRESS grant, Novel Quantum Algorithms from Fast Classical Transforms. K.F. and J.N. are supported by Northeastern University.

\appendix

\bibliographystyle{JHEP}
\bibliography{main}

\end{document}